\documentclass[11pt]{article}
\usepackage{amssymb}
\usepackage{graphicx}
\usepackage{amsmath}
\usepackage[T1]{fontenc}
\usepackage{amsmath, amssymb, graphicx, multicol, color, epic, eepicemu, epsfig,appendix, float,enumerate}

\newtheorem{proposition}{Proposition}[section]

\newenvironment{proof}[1][Proof]{\textbf{#1.} }{\ \rule{0.5em}{0.5em}}
\def\R{{\bar{R}}}
\def\X{{\cal X}}
\def\Y{{\cal Y}}
\def\Z{{\cal Z}}
\def\e{\epsilon}

\begin{document}
\pdfoutput=1
\title{Hybrid classifiers of pairwise Markov models}
\author{Kristi Kuljus\footnote{University of Tartu, Estonia; e-mail: kristi.kuljus@ut.ee}, J\"{u}ri Lember\footnote{University of Tartu, Estonia;  e-mail: jyri.lember@ut.ee}}
\maketitle
\begin{abstract} The article studies segmentation problem (also known as classification problem) with pairwise Markov models (PMMs). A PMM is a process where the observation process and underlying state sequence form a two-dimensional Markov chain, it is a natural generalization of a hidden Markov model. To demonstrate the richness of the class of PMMs, we examine closer a few examples of rather different types of PMMs: a model for two related Markov chains, a model that allows to model an  inhomogeneous Markov chain as a homogeneous one and a semi-Markov model. The segmentation problem assumes that one of the marginal processes is observed and the other one is not, the problem is to estimate the unobserved state path given the observations. The standard state path estimators often used are the so-called Viterbi path (a sequence with maximum state path probability given the observations) or the pointwise maximum a posteriori (PMAP) path (a sequence that maximizes the conditional state probability for given observations pointwise). Both these estimators have their limitations, therefore we derive formulas for calculating the so-called hybrid path estimators which interpolate between the PMAP and Viterbi path. We apply the introduced algorithms to the studied models in order to demonstrate the properties of different segmentation methods, and to illustrate large variation in behaviour of different segmentation methods in different PMMs. The studied examples show that a segmentation method should always be chosen with care by taking into account the particular model of interest.

\end{abstract}
\section{Introduction}
\subsection{Pairwise Markov models}
Let $\X$ and $\Y$ be discrete sets and let
$\{Z_t\}_{t=1}^{\infty}=\{(X_t,Y_t)\}_{t=1}^{\infty}$ be a
homogeneous Markov chain taking values in $\Z\subset \X\times \Y$.
Here the state space  $\Z$  can be a proper subset of $\X\times \Y$. Following the terminology proposed by W. Piezcynski (see, e.g. \cite{P03,P04,Pzebra,Psemi, P12, P13,gorynin18}), we call the
process $Z=(X,Y)$ a {\it pairwise Markov chain} or a {\it pairwise
Markov model (PMM)}. The name reflects the fact that although the
processes $X$ or $Y$ might lack the Markov property, conditionally
on $X$ (or on $Y$) the process $Y$ (or $X$) is an inhomogeneous Markov
chain (see Proposition 2.1 in \cite{P03}). It turns out that the two-dimensional structure makes PMMs very useful and flexible allowing to consider many stochastic models as a homogeneous Markov chain. In Section \ref{sec:examples}, we present a few examples of rather different PMMs. The first example in Subsection \ref{sec:milano} presents a  parametric class of PMMs, where both marginal processes are Markov chains with given transition matrices and the parameters allow to model dependence structure between the marginal chains.  The property that both marginal processes are Markov chains is rather untypical for PMMs, because usually at least one of the marginal processes does not have the Markov property. However, incorporating two Markov chains into one  might be useful and in many respects that model is very special.
Our second example in Subsection \ref{sec:regime}  -- a {\it regime-switching model} -- allows to consider an inhomogeneous Markov chain as a PMM (and hence as a homogeneous chain). In particular, suppose that $X$ is a stochastic process that in a certain random time-period behaves as a  homogeneous Markov chain, but then the transition matrix changes. After the change, $X$ evolves again as a Markov chain but now with another transition matrix, and after a certain random time-period the matrix changes again. Such a model can be considered as a PMM $(X,Y)$, where the $Y$-process is a Markov chain that governs the time periods -- {\it regimes} -- for different transition matrices and the $X$-process is the observed one. We argue that given a realization of $Y$, $X$ is an inhomogeneous Markov chain, but unconditionally it lacks the Markov property. The third example in Subsection \ref{sec:semi-markov} allows to consider a semi-Markov process as a PMM and the fourth example in Subsection \ref{sec:semi-regime} combines the regime-switching model and semi-Markov model into one. In the regime-switching model the regime process $Y$ is a Markov chain, thus the times $Y$ spends in a particular regime are geometrically distributed. Replacing $Y$ by a semi-Markov PMM, let it be $(U,V)$, allows us to generalize the regime-switching model so that the inter-regime times don't have necessarily geometrical distributions. Hence the resulting model $(X,(U,V))$ is actually a three-dimensional Markov chain known as a {\it triplet Markov model}, see e.g. \cite{P04b,P07,Pzebra,Psemi, P12,gorynin18}.
\subsection{Segmentation problem}
Suppose that a researcher has a realization $x_1,\ldots,x_n$ of observations $X_1,\ldots,X_n$ and the objective is to estimate the unobserved class variables $Y_1,\ldots,Y_n$. We shall call estimation of unobserved class variables a {\it segmentation problem}  (also known as classification, denoising or decoding). For dependent observations a classical latent variable model often considered in unsupervised learning is a {\it hidden Markov model} (HMM), where $Y$ is a Markov chain and given $Y_1,\ldots,Y_n$, the observations $X_1,\ldots,X_n$ are conditionally independent. In HMMs, the distribution of $X_t$ depends solely on $Y_t$ (sometimes the dependence on $Y_{t-1}$ is allowed as well).
As the examples in Section \ref{sec:examples} illustrate, the class of PMMs is much larger compared to HMMs, allowing also for conditionally dependent observations and for $Y$ that is not a Markov chain. Therefore, in this article we study different segmentation methods for this rich class of models. \\\\
In Section \ref{sec:segmentation} we briefly recall the decision-theoretical foundations of segmentation theory. The standard solutions of the segmentation problem are either a state sequence that maximizes the conditional probability $P(Y_1=y_1,\ldots,Y_n=y_n|X_1=x_1,\ldots,X_n=x_n)$ over all sequences $(y_1,\ldots,y_n)$ -- the so-called {\it Viterbi path} -- or a sequence which maximizes the probability $P(Y_t=y|X_1=x_1,\ldots,X_n=x_n)$ over all possible states $y$ for every $t\in\{1,\ldots,n\}$ separately. This pointwise estimator is called the {\it pointwise maximum a posteriori (PMAP) path}. Since the Viterbi and forward-backward algorithms apply for any PMM, both the Viterbi and PMAP path can be easily found. However, both these solutions have their limitations. The PMAP path is guaranteed to maximize the expected number of correctly estimated classes, but it might have a zero (conditional) probability, hence it can be {\it inadmissible}. The Viterbi path, on the other hand, might be rather inaccurate in terms of pointwise classification errors. These deficiencies are well known in the literature and were pointed out already by L. Rabiner in his seminal tutorial \cite{rabiner}. As a remedy against inadmissible PMAP paths, he proposed to replace the PMAP path with a state path that maximizes the expected number of correctly estimated state blocks of length $k$, $k\geq 2$, we shall call these state path estimators {\it Rabiner $k$-block paths}. Since inadmissibility is mostly caused due to impossible transitions in the model and any impossible transition automatically results in a wrongly estimated block, it is natural to hope that any $k$-block path should minimize such impossible  transitions, especially for larger $k$. However, Rabiner $k$-block paths can still be inadmissible and this might happen quite easily. Often various counterexamples consider two-block paths, but our example in Subsection \ref{sec:example3} shows that even $5$-block paths can be inadmissible. \\\\
To deal with the problem of inadmissibility more efficiently, a family of {\it hybrid path estimators} was defined in \cite{seg}. The idea behind the hybrid paths mimics partly the Rabiner $k$-block idea, but instead of block length $k$, a hybrid path depends on a regularization  parameter $C\geq 0$ and any hybrid path is guaranteed to be admissible. We shall see in Section \ref{sec:segmentation} that when $C=k-1$, then for smaller  integer values $k$ the corresponding hybrid path can be considered as an analogue of the $k$-block path. Hybrid paths interpolate between PMAP and Viterbi paths: for $C=0$ the hybrid path equals the PMAP path and from certain $C$ on the hybrid path equals the Viterbi path. This means that there exists a constant $C_o$ such that the hybrid path becomes Viterbi when $C>C_o$. The same holds for the Rabiner $k$-block paths: when $k\geq n$, there is only one block to optimize and the solution is Viterbi. By intuition one might expect the same property to hold for hybrid paths --  that for $C>n$ any hybrid path is Viterbi. This turns out to be very wrong, because for hybrid paths the block interpretation does not work for large integers $k$, besides the critical parameter $C_o$ depends both on model and observations. In Subsection \ref{sec:example1} we present an example showing the extreme instability of $C_o$. In that example we can see that adding just three more observations to a particular observation sequence might increase $C_o$ unboundedly. But on the other hand,  for the same model a typical $C_o$ estimated with generated data is very small. Another study in Subsection \ref{sec:example3} shows that $C_o$ might be much larger than the sample size $n$. All examples in Section \ref{sec:simulations} demonstrate that a segmentation method should be chosen with care by taking into account the particular model studied. \\\\
In \cite{seg}, it was shown that for HMMs the hybrid paths can be easily found by a dynamic programming algorithm that combines the Viterbi algorithm and the forward-backward recursions. In Proposition \ref{prop1} we show that the same algorithm holds for PMMs so that all the hybrid paths can be found with complexity $O(|\Y|n)$, the complexity is independent of $C$. The algorithm for the Rabiner $k$-block paths has complexity $O(|\Y|^{k-1}n)$ (Proposition \ref{prop2}). This means that for bigger blocks Rabiner $k$-block paths are very time-consuming to calculate. To illustrate the performance of hybrid and Rabiner $k$-block paths, in Section \ref{sec:simulations} both algorithms are applied to data generated from different PMMs.
\\\\
Throughout the article we assume that the observation alphabet $\X$ is discrete. The reason for making this assumption is to reduce mathematical technicalities. In real segmentation problems $\X$ is often uncountable, say $\X=\mathbb{R}^d$. All ideas and algorithms of this note carry on to the uncountable $\X$ as well, just the notations and proofs would be more technical. For formal definition of PMMs, the  Viterbi algorithm and related concepts in the case of general $\X$, see \cite{sova1,sova2,sova3}.
%
%
%
\section{Preliminaries}
Recall that $\X$ and $\Y$ are discrete sets and
$\{Z_t\}_{t=1}^{\infty}=\{(X_t,Y_t)\}_{t=1}^{\infty}$ is a homogeneous Markov chain taking values in $\Z\subset \X\times \Y$.
In what follows, we shall denote by $p$ various probabilities like
transition probabilities
$p(z_{t+1}|z_t)=p(x_{t+1},y_{t+1}|x_t,y_t)=P(X_{t+1}=x_{t+1},Y_{t+1}=y_{t+1}|X_{t}=x_{t},Y_{t}=y_{t})$,
joint distributions $p(x_s,\ldots,x_t)=P(X_s=x_s,\ldots,X_t=x_t)$
etc. We use abbreviation $x_s^t=(x_s,\ldots,x_t)$ and when $s=1$, we write $x^t$ instead of $x_1^t$. We denote
$p_t(y|x^n)=P(Y_t=y|X^n=x^n)$.  As usual in a discrete setting, any conditional probability implies that the probability of the condition is strictly positive. Sometimes we abuse the notation a bit by writing
$p(x_{t+1}=x|x_t=x')$ instead of $P(X_{t+1}=x|X_t=x')$ etc.
\\\\
Pairwise Markov models is a large class of stochastic models that can be classified via the properties of transition probabilities $p(z_{t+1}|z_t)$. When the transition probabilities factorize as
$$p(z_{t+1}|z_t)=p(y_{t+1}|y_t)p(x_{t+1}|y_{t+1}),$$
then we have an HMM. Clearly  HMMs is a very narrow subclass of PMMs.
A broader  subclass of PMMs is the class of {\it Markov switching models}
(see \cite{HMM-book}), where the observations are not conditionally independent any more:
$$p(z_{t+1}|z_t)=p(y_{t+1}|y_t)p(x_{t+1}|y_{t+1},x_{t}).$$
Thus, HMMs is a special case of Markov switching models. A more general class of models where $Y$ is a Markov chain is the class of {\it hidden Markov models with dependent noise} (HMM-DN), see \cite{P13}, where the transition probabilities
factorize as follows:
\begin{equation}\label{HMM-DN}
p(z_{t+1}|z_t)=p(y_{t+1}|y_t)p(x_{t+1}|y_{t+1}, y_t,x_{t}).\end{equation}
Hence the class of Markov switching models is a special class of HMM-DNs.
Observe that  (\ref{HMM-DN}) is  equivalent to
\begin{equation}\label{HMM-DN2}p(y_{t+1}|z_t)=p(y_{t+1}|y_t).
\end{equation}
Since (\ref{HMM-DN}) is an important property, we shall examine it a bit closer. Suppose
${\cal Z}={\cal X}\times {\cal Y}$. Let $|{\cal X}|=k\leq \infty$ and $|{\cal Y}|=l\leq \infty$  and let the $kl$ elements of ${\cal X}\times {\cal Y}$ be ordered as follows:
$$\X\times \Y=\{(\chi_1,\gamma_1),\ldots,(\chi_k,\gamma_1),(\chi_1,\gamma_2),\ldots,(\chi_k,\gamma_2),\ldots,(\chi_1,\gamma_l),\ldots, (\chi_k,\gamma_l)\}.$$
Then $Z$ is an HMM-DN if and only if the $kl\times kl$ transition matrix factorizes as follows:
\begin{equation}\label{maatriks}\left(
  \begin{array}{cccc}
    p_{11}\cdot A_{11}& p_{12}\cdot A_{12} & \ldots & p_{1l}\cdot A_{1l} \\
     p_{21}\cdot A_{21}& p_{22}\cdot A_{22} & \ldots & p_{2l}\cdot A_{2l} \\
    \cdots & \cdots & \cdots & \cdots \\
     p_{l1}\cdot A_{l1}& p_{l2}\cdot A_{l2} & \ldots & p_{ll}\cdot A_{ll} \\
  \end{array}
\right),
\end{equation}
where  $p_{ij}=P(Y_{t+1}=\gamma_j|Y_t=\gamma_i)$ and $A_{ij}$ are the following transition matrices:
$$A_{ij}=\left(
  \begin{array}{ccc}
    a^{(ij)}_{11} & \cdots  & a^{(ij)}_{1k} \\
    \cdots & \cdots & \cdots \\
     a^{(ij)}_{k1} & \cdots &  a^{(ij)}_{kk} \\
  \end{array}
\right), \quad a^{(ij)}_{u v}=P(X_{t+1}=\chi_v|Y_{t+1}=\gamma_j,Y_{t}=\gamma_i, X_t=\chi_u).$$
If $Z$ is a Markov switching model, then the probability $a_{uv}^{(ij)}$ is independent of $i$ and then in (\ref{maatriks}),
$$A_{ij}=A_j, \quad A_j=(a^{(j)}_{uv}),\quad a^{(j)}_{uv}=P(X_{t+1}=\chi_v|Y_{t+1}=\gamma_j, X_t=\chi_u).$$
When $Z$ is an HMM, then $a^{(ij)}_{uv}$ is independent of both $i$ and $u$ implying that all the rows in $A_j$ are equal.\\\\
The property in (\ref{HMM-DN2}) that defines HMM-DNs  and holds for all their subclasses has important implications. Since it obviously implies that
$p(y_{t+1}^n|z_t)=p(y_{t+1}^n|y_t)$, we see that the conditional distribution of $x_t$ given the whole sequence $y^n$ depends only on $y^t$:
\begin{equation}\label{alfa}
p(x_t|y^n)={p(x_t,y^t)p(y_{t+1}^n|x_t,y_t)\over p(y^t)p(y_{t+1}^n|y_t)}={p(x_t,y^t)p(y_{t+1}^n|y_t)\over p(y^t)p(y_{t+1}^n|y_t)} = p(x_t|y^t).\end{equation}
If the model happens to be such that (\ref{HMM-DN2}) holds for the time-reversed chain, that is $p(y_{t-1}|z_t)=p(y_{t-1}|y_t)$, then clearly $p(x_t|y^n)=p(x_t|y_t^n)$.
Therefore, if the model is HMM-DN and the time-reversed chain is HMM-DN as well, then it must hold that
$p(x_t|y^t)=p(x_t|y^n)=p(x_t|y_t^n)$, implying that $p(x_t|y^n)=p(x_t|y_t)$. Hence the conditional distribution of $x_t$ given the whole sequence $y^n$ depends solely on $y_t$. Any HMM has this particular property, but in the case of HMMs in addition the $X$-variables are conditionally independent of $Y$, that is $p(x^n|y^n)=\prod_{t=1}^n p(x_t|y_t)$. Another implication of (\ref{HMM-DN2}) is that the conditional transition probabilities $p(x_{t+1}|x_t,y^n)$ depend only on $y_t$ and $y_{t+1}$, but not on $y_{t+2}^n$. Indeed, when (\ref{HMM-DN2}) holds, then
\begin{equation}\label{tr2}
p(x_{t+1}|x_t,y^n)={p(x_{t+1},x_t,y^n)\over p(x_t,y^n)}={p(x_t,y^t)p(z_{t+1}|z_t)p(y_{t+2}^n|y_{t+1})\over p(x_t,y^t)p(y_{t+1}|z_t)p(y_{t+2}^n|y_{t+1},z_t)}=p(x_{t+1}|x_t,y_t,y_{t+1}),
\end{equation}
because  $p(y_{t+2}^n|y_{t+1},z_t)=\sum_{x_{t+1}}p(y_{t+2}^n|x_{t+1},y_{t+1},z_t)p(x_{t+1}|y_{t+1},z_t)=p(y_{t+2}^n|y_{t+1}).$ Without (\ref{HMM-DN2}) this probability might depend on $y_{t+2}^n$, i.e. for any PMM,   $p(x_{t+1}|x_t,y^n)=p(x_{t+1}|x_t,y_t^n)$.
\paragraph{Preserving the Markov property.}
A question raised already in the very first papers about PMMs \cite{P03}
was when is the $Y$-process unconditionally a Markov chain. If it is, then a PMM can be considered as a direct generalization of HMMs.
It is easy to see that when $Z$ is an HMM-DN, then $Y$ is a Markov chain. Indeed, under (\ref{HMM-DN2}) we have for any $t\geq 2$:
$$p(y_{t+1}|y^t)=\sum_{x_t\in \X}p(y_{t+1}|x_t,y^t)p(x_t|y^t)=\sum_{x_t\in \X}p(y_{t+1}|x_t,y_t)p(x_t|y^t)=\sum_{x_t\in \X}p(y_{t+1}|y_t)p(x_t|y^t)=p(y_{t+1}|y_t).$$
Here the second equality follows from the Markov property and the third equality follows from (\ref{HMM-DN2}). Hence the matrix $(p_{ij})$ in  representation (\ref{maatriks}) is the transition matrix of the Markov chain $Y$.
Since a Markov switching model is a special case of HMM-DNs, then under that model $Y$ is a Markov chain as well. It has been an open question whether being an HMM-DN is also a necessary property for $Y$ being a Markov chain.
The following proposition shows that under a special condition it is indeed so.
 \begin{proposition}\label{propHMMDN} Let $t$ be fixed and let  $(y_{t-1},y_t,y_{t+1})$ be a state triplet such that  $y_{t+1}=y_{t-1}$ and $p(x_t|y_t,y_{t-1})=p(x_t|y_t,y_{t+1})$ for every $x_t$.
 Then the Markov property $p(y_{t+1}|y^t)=p(y_{t+1}|y_t)$ implies that (\ref{HMM-DN}) holds.
\end{proposition}
\begin{proof} Take $(y_{t-1},y_t,y_{t+1})$ such that $p(y_{t-1},y_t,y_{t+1})>0$ and $y_{t+1}=y_{t-1}$. The Markov property implies $p(y_{t+1}|y_t,y_{t-1})=p(y_{t+1}|y_t)$. Then
\begin{equation*}
p(y_{t+1}|y_t,y_{t-1})=\sum_{x_t\in \X} p(y_{t+1}|x_t,y_{t},y_{t-1})p(x_t|y_{t},y_{t-1})=\sum_{x_t\in \X} p(y_{t+1}|x_t,y_{t})p(x_t|y_{t},y_{t-1})
\end{equation*}
\begin{equation*}
=\sum_{x_t\in \X} p(y_{t+1}|y_{t}){p(x_t|y_t,y_{t+1})\over p(x_t|y_t)}
p(x_t|y_{t},y_{t-1})=p(y_{t+1}|y_{t})\sum_{x_t\in \X}{p(x_t|y_t,y_{t+1})\over p(x_t|y_t)}p(x_t|y_{t},y_{t-1}).
\end{equation*}
Hence, the equality $p(y_{t+1}|y_t,y_{t-1})=p(y_{t+1}|y_t)$ implies that
\begin{equation}\label{int}\sum_{x_t\in \X}{p(x_t|y_t,y_{t+1})\over p(x_t|y_t)}p(x_t|y_{t},y_{t-1})=1.\end{equation}
By the assumption, $p(x_t|y_t,y_{t+1})=p(x_t|y_{t},y_{t-1})$ for every $x_t$. We now show the following implication: if $p$ and $q$ are probability measures on $\X$ so that
$$\sum_{x\in \X} {p(x)\over q(x)}p(x)=1,$$
then $p(x)=q(x)$ for every $x$.  It follows from Jensen's inequality that
\begin{equation}\label{j}
\ln \Big(\sum_{x\in \X} {p(x)\over q(x)}p(x)\Big)\geq \sum_{x\in \X} \ln \big({p(x)\over q(x)}\big)p(x),\end{equation}
the equality holds if and only if $p(x)=q(x)$ for every $x\in \X$. The right-hand side of (\ref{j}) is the Kullback-Leibler divergence between $p$ and $q$, hence non-negative. The left-hand side is 0 by the assumption. Thus it follows that $p=q$. Therefore also $p(x_t|y_t,y_{t+1})=p(x_t|y_t)$ for every $x_t$ and multiplying both sides by $p(y_t,y_{t+1})$ gives
$$p(y_{t+1}|x_t,y_t)p(x_t,y_t)=p(x_t,y_t,y_{t+1})=p(x_t|y_t)p(y_t,y_{t+1})=p(x_t,y_t)p(y_{t+1}|y_t),$$
thus (\ref{HMM-DN}) holds.
\end{proof}\\\\
Proposition \ref{propHMMDN} is Proposition 2.3 in \cite{P03}, but the proof there does not use Jensen's inequality and is therefore more complicated. The assumptions of Proposition \ref{propHMMDN} are satisfied when $Z$ is a stationary and reversible Markov chain. Thus, for a stationary and reversible $Z$, the marginal process $Y$ is Markov if and only if $Z$ is an HMM-DN. However, the assumption of Proposition \ref{propHMMDN} might hold also when $Z$ is not reversible and  we shall present an example of a PMM where $Y$ is a Markov chain, but $Z$ is not an HMM-DN (Subsection \ref{sec:milano}). This  example shows that the HMM-DN property is not necessary for $Y$ being a Markov chain.
\section{Some examples of discrete PMMs} \label{sec:examples}
\subsection{The related Markov chains model}\label{sec:milano}
 Usually in pairwise Markov models the $X$-process is not a Markov chain even if $Y$ is. Indeed, even for an HMM, the observation process $X$ has a long memory and is typically not a Markov chain. The PMM considered in this section is, however, deliberately constructed so that besides $Y$ also the $X$-chain were a Markov chain. In particular, the model is an HMM-DN (so that $Y$ is a Markov chain) and it is also an HMM-DN when the roles of $X$ and $Y$ are changed. In what follows, the latter property shall be referred to as {\it  HMM-DN by $X$} and it implies that $X$ is a Markov chain as well. For simplicity we consider the two-letter alphabets $\X=\{1,2\}$ and $\Y=\{a,b\}$. Let $P_X$ and $P_Y$ be the transition matrices of $X$ and $Y$, respectively (all entries positive):
$$P_X= \left(
    \begin{array}{cc}
      p & 1-p \\
      q & 1-q \\
    \end{array}
  \right), \quad
   P_Y= \left(
    \begin{array}{cc}
      p' & 1-p' \\
      q' & 1-q' \\
    \end{array}
  \right).$$
Given these matrices, we construct a parametric family of PMMs with the state space $\X\times\Y$ so that the marginal processes $X$ and $Y$ were both Markov chains with these transition matrices. The parameters allow to tune the dependence between the $X$- and $Y$-sequence. In particular, a certain combination of parameters will yield the case where $X$ and $Y$ are independent Markov processes, and another combination will
provide the maximal dependence case. A PMM fulfilling these requirements has the following transition matrix:
\begin{equation}\label{milano}
\mathbb{P} =
\bordermatrix{ ~ & (1,a) & (1,b) & (2,a) & (2,b)\cr
     (1,a) & p\lambda_1 & p(1-\lambda_1) & p'-p \lambda_1 & 1+p\lambda_1 -p'-p\cr
     (1,b) & p\lambda_2 & p(1-\lambda_2) & q'-p\lambda_2 & 1+p \lambda_2 -q'-p\cr
     (2,a) & q\mu_1 & q(1-\mu_1) & p'-q\mu_1 & 1+q\mu_1-p'-q \cr
     (2,b) & q\mu_2 & q(1-\mu_2) & q'-q\mu_2 & 1+q\mu_2-q'-q\cr
    },
\end{equation}
where $\lambda_i,\mu_i, i=1,2$ are parameters satisfying the following conditions:
\begin{align*}
&\lambda_1 \in \left[\max\left\{{p'+p-1\over p},0\right\}, \min\left\{{p' \over p},1\right\} \right],
  \quad \lambda_2 \in \left[ \max\left\{ {q'+p-1\over p}, 0 \right\},\min\left\{{q' \over p}, 1\right\} \right], \\
 &\mu_1\in \left[\max\left\{{p'+q-1\over q}, 0\right\},\min\left\{{p'\over q},1 \right\} \right],
\quad \mu_2 \in \left[\max\left\{{q'+q-1\over q}, 0\right\}, \min\left\{{q'\over q}, 1 \right\} \right]. \label{MuCond}
\end{align*}
The conditions above ensure that all probabilities are in $[0,1]$.
The transition matrix in (\ref{milano}) is of form  (\ref{maatriks}) with $(p_{ij})=P_Y$, and after changing the roles of $X$ and $Y$ it has the same form with $(p_{ij})=P_X$.  Hence $(X,Y)$ is an HMM-DN as well as HMM-DN by $X$. Therefore, $X$ and $Y$ are both Markov chains with transition matrices $P_X$ and $P_Y$. The parameters have the following meaning:
\begin{align*}
&\lambda_1=P(Y_{t+1}=a|Y_t=a, X_t=1, X_{t+1}=1),\quad \lambda_2=P(Y_{t+1}=a|Y_t=b, X_t=1, X_{t+1}=1),\\
& \mu_1=P(Y_{t+1}=a|Y_t=a, X_t=2, X_{t+1}=1),\quad \mu_2=P(Y_{t+1}=a|Y_t=b, X_t=2, X_{t+1}=1).
\end{align*}
Due to the representation in (\ref{maatriks}) it is clear that any  PMM that is an HMM-DN and also HMM-DN by $X$ must have a transition matrix like (\ref{milano}).
However, it does not necessarily mean that the transition matrix in (\ref{milano}) is the only possibility for both $X$ and $Y$ being Markov chains, because the marginal processes might be Markov chains even if the model is not an HMM-DN and HMM-DN by $X$. As we shall see, an example of such a model is the (time) reversed model of $Z$ with the transition matrix in (\ref{milano}). Indeed, reversing the time does not change the Markov property of marginal processes, but it might spoil the HMM-DN-properties. Let us remark that sometimes it is useful to rewrite the last two columns of (\ref{milano}) by introducing artificial parameters $\theta_i$ and $\rho_i$, $i=1,2$, see the appendix.
\\\\
The PMMs with transition matrix (\ref{milano}) were introduced in \cite{zucca} to model dependencies between two-state Markov chains, see also \cite{sova1}. The model can be generalized to the case with bigger alphabets, for such a generalization with $P_X=P_Y$, see \cite{avans}. Observe that when $\lambda_1=\mu_1=p'$, $\lambda_2=\mu_2=q'$ and the initial distribution factorizes as $\pi(1,a)=\pi_X(1)\pi_Y(a)$, then $X$ and $Y$ are independent Markov chains, because then  $p(x_{t+1},y_{t+1}|x_t,y_t)=p(x_{t+1}|x_t)p(y_{t+1}|y_t).$ In what follows, we shall consider closer the case with $p'=p$ and $q'=q$, i.e. $P_X=P_Y$. Then taking $\lambda_1=\mu_2=1$ and choosing initial distribution $\pi$ so that $\pi(1,a)=\pi(2,b)=1$, we get the maximal dependence between $X$ and $Y$: $Y_t=a$ if and only if $X_t=1$ for any $t\geq 1$. Similarly, when
$q=1-p$, the choice $\lambda_2=\mu_1=0$ with $\pi(1,b)=\pi(2,a)=1$ yields the other case of maximal dependence: $X_t=1$ if and only if $Y_t=b$.
\paragraph{Reversibility.} Let us consider the case with $p=p'\in (0,1)$ and $q=q'\in (0,1)$, i.e. $P_X=P_Y$. We aim to find the conditions that ensure the reversibility of  $Z$. Recall that any two-state Markov chain with positive transition probabilities is always reversible. Thus, for the time-reversed $Z$, the marginal processes $X$ and $Y$ remain to be Markov chains with the same transition matrices $P_X$ and $P_Y$. Suppose now that the reversed chain is an HMM-DN and also HMM-DN by $X$. Then, as noted above, it must have the transition matrix as in (\ref{milano}) with some parameters $\lambda^R_i$ and $\mu^R_i$, $i=1,2$. The elements on the main diagonal of the  original and reversed transition matrix always coincide. Therefore, if the reversed chain has a transition matrix of the form in (\ref{milano}), then all its parameters should be the same as the parameters of the original chain, i.e. $\lambda^R_i=\lambda_i$ and $\mu^R_i=\mu_i, i=1,2.$ In other words, our model is reversible if and only if it is an HMM-DN as well as HMM-DN by $X$. The necessary and sufficient conditions for that property to hold are the following:
\begin{equation} \label{Yequations}
p_{11}^R+p_{13}^R=p, \quad p_{31}^R+p_{33}^R=p, \quad p_{21}^R+p_{23}^R=q, \quad p_{41}^R+p_{43}^R=q,
\end{equation}
\begin{equation} \label{Xequations}
p_{11}^R+p_{12}^R=p, \quad p_{21}^R+p_{22}^R=p, \quad p_{31}^R+p_{32}^R=q, \quad p_{41}^R+p_{42}^R=q.
\end{equation}
In the formulas above $(p^R_{ij})$ stands for the transition matrix of the reversed chain, the state space $\X\times Y$ is encoded as $\{1,2,3,4\}$ (i.e. $(1,a)=1$, $(1,b)=2$, $(2,a)=3$, $(2,b)=4$). The equalities in (\ref{Yequations}) are necessary and sufficient for the reversed chain being an HMM-DN and the equalities in (\ref{Xequations}) are necessary and sufficient for the reversed $Z$ being an HMM-DN by $X$. For calculating $(p^R_{ij})$, one needs to know the unique stationary distribution $(\pi_1,\pi_2,\pi_3,\pi_4)$ of $Z$ that for our model is given by
\begin{align}\notag
&\pi_1=\pi(1,a)=\pi(1)c,\quad \pi_2=\pi(1,b)=\pi(1)(1-c)=\pi(2,a)=\pi_3,\quad \pi_4=\pi(2,b)=1-\pi(1)(2-c),\\\label{st}
& \mbox{with} \, \, \pi(1)={q\over 1-p+q},\quad \mbox{and} \quad c=\pi(a|1)={ p(\lambda_2-\mu_2)+q(\mu_1-\mu_2)+\mu_2\over q(\mu_1-\mu_2)+p(\lambda_2-\lambda_1)+1}.
\end{align}
The distribution in (\ref{st}) is the unique stationary distribution unless $p+q=1$ and $\lambda_1=\mu_2=1,\lambda_2=\mu_1=0$.
With  $p_{ij}^R=(\pi_j/\pi_i)p_{ji}$, we can now verify that all equations in (\ref{Yequations}) hold if and only if
\begin{equation} \label{condYrev}
c=\frac{q\mu_1}{q\mu_1+p(1-\lambda_1)},
\end{equation}
and  all equations in $(\ref{Xequations})$ hold if and only if
\begin{equation} \label{condXrev}
c=\frac{\lambda_2}{\lambda_2-\lambda_1+1}.
\end{equation}
It can be concluded that (\ref{condYrev}) and (\ref{condXrev}) are necessary and sufficient conditions for the time-reversed chain being an HMM-DN and HMM-DN by $X$; or equivalently that these are necessary and sufficient conditions for the reversibility of $Z$. An example of a set of parameters such that (\ref{condYrev}) and (\ref{condXrev}) hold is for example: $p=0.55,q=0.8,\lambda_1=0.3, \lambda_2=0.65,\mu_1=0.446875,\mu_2=0.8894737$.
\\\\ Suppose now that (\ref{condYrev}) fails. Then the reversed chain is not an HMM-DN, but the marginal process $Y$ of the reversed chain is still a Markov process. Hence we have an example of a PMM such that $Y$ is a Markov chain, but the model is not an HMM-DN. For an example of such a chain consider (\ref{milano}) with $p=0.55$, $q=0.8$, $\lambda_1=0.52$, $\lambda_2=0.8$, $\mu_1=0.6$, $\mu_2=0.9$. The transition matrices of the original and reversed chains are
\begin{equation*}\mathbb{P} =\left(
                  \begin{array}{cccc}
                  0.286 & 0.264 & 0.264  & 0.186\\
                  0.440 & 0.110 & 0.360  & 0.090\\
                  0.480 & 0.320 & 0.070 & 0.130\\
                  0.720 & 0.080 & 0.080 & 0.120
                  \end{array}
                \right), \mathbb{P}^R = \left(
                               \begin{array}{cccc}
                                 0.2860000 & 0.2247273 & 0.2451570 & 0.2441157\\
                                 0.5168932 & 0.1100000 & 0.3200000 & 0.0531068\\
                                 0.5168932 & 0.3600000 & 0.0700000 & 0.0531068\\
                                 0.5485923 & 0.1355759 & 0.1958318 & 0.1200000
                               \end{array}
                             \right).\end{equation*}
In this example both conditions (\ref{condYrev}) and (\ref{condXrev}) fail and one can see that the reversed chain is indeed neither HMM-DN nor HMM-DN by $X$. However, both marginal chains remain Markov chains with the same transition matrices ($p=0.55$, $q=0.8$). This example shows that the marginal process might be a Markov chain without the HMM-DN property.
%
\subsection{Regime-switching model}\label{sec:regime} The model studied in the present subsection allows to present an inhomogeneous Markov chain as a  homogeneous PMM. Such models have been successfully used in segmentation of non-stationary images \cite{P04b,Pzebra,P12}.  Suppose we would like to model the process $X$ which during certain random time-periods evolves as a homogeneous Markov chain with different transition matrices. Thus, the transition matrices of $X$ can be different in different periods, but since these periods are random, the process lacks the Markov property. Suppose we consider three different transition matrices $P_A$, $P_B$ and $P_C$ on a state space $\X$. The process $X$ starts as a Markov chain with one of the three matrices, then after a certain random time-period the transition matrix changes and during the next time period the process evolves as a Markov chain with the new transition matrix. After a certain time period the matrix changes again and so on. To consider such a model as a homogeneous Markov chain, we embed $X$ into a PMM $(X,Y)$, where the $Y$-process takes values in  $\Y=\{A,B,C\}$ and we interpret $A,B,C$ as {regimes}. Inside a current regime $a\in \Y$, the transition probability $x_t\to x_{t+1}$ is  determined by the matrix $P_a$: for every $y^n\in \Y^n$ such that $y_t=y_{t+1}=a$ and $\forall i,j\in \X$,
\begin{equation}\label{regime}
P(X_{t+1}=j|X_t=i,Y^n=y^n)=P(X_{t+1}=j|X_t=i,Y_t=a,Y_{t+1}=a)=P_a(i,j).
\end{equation}
Thus, $P_A$, $P_B$ and $P_C$ are the transition matrices on $\X \times \X$ inside every regime defined by (\ref{regime}). Let $P_Y$ denote the transition matrix of regimes:
\[ P_Y=\left(
  \begin{array}{ccc}
    r_A & r_{AB} & r_{AC} \\
    r_{BA}& r_B & r_{BC} \\
    r_{CA} & r_{CB} & r_C \\
  \end{array}
\right).
\]
Since (\ref{regime}) is the same as (\ref{tr2}), we know that for (\ref{regime}) to hold it suffices to construct $(X,Y)$ so that it would be an HMM-DN, i.e. the transition matrix of $(X,Y)$ must factorize as in (\ref{maatriks}). Thus, the transition matrix of $(X,Y)$ must be as follows:
$$ \mathbb{Q} = \left(
  \begin{array}{ccc}
    r_A P_A & r_{AB} P_{AB} & r_{AC} P_{AC} \\
    r_{BA} P_{BA} & r_B P_B & r_{BC} P_{BC} \\
    r_{CA} P_{CA} & r_{CB} P_{CB} & r_C P_{C}\\
  \end{array}
\right),
$$
where the matrices $P_{AB}$, $P_{AC}$, $P_{BA}$, $P_{BC}$, $P_{CA}$ and $P_{CB}$ define the transitions of the $X$-process when the regime changes. For example,
$P_{AB}(i,j)=P(X_{t+1}=j|X_t=i,Y_t=A,Y_{t+1}=B).$ We shall call these matrices {\it inter-regime matrices}. In practice, the process is often  supposed to stay in the same regime for quite a long time, thus the off-diagonal elements of the regime transition matrix $P_Y$ are close to zero \cite{Pzebra}. In this case the choice of the inter-regime matrices has a little influence, but they should be specified and in principle there are infinitely many possibilities for doing it. On the other hand, there is only one way to choose the inter-regime matrices so that the overall PMM were a  Markov switching model, namely by choosing $P_{AB}=P_{CB}=P_B$, $P_{BA}=P_{CA}=P_A$, $P_{AC}=P_{BC}=P_C$. With this choice, the transition probability $i\to j$ under regime change is specified by the new regime, that is for every $a\in \Y$,
$$P_{aB}(i,j)=P(X_{t+1}=j|X_t=i,Y_t=a,Y_{t+1}=B)=P(X_{t+1}=j|X_t=i,Y_t=B,Y_{t+1}=B),$$
and we can see that inter-regime transitions are indeed independent of $Y_t$.
Of course, one can choose the inter-regime matrices so that the old regime specifies transitions ($P_{AB}=P_{AC}=P_A$, $P_{BA}=P_{BC}=P_B$,  $P_{CA}=P_{CB}=P_C$) or for example so that all inter-regime matrices are equal; there are  many other  meaningful options.  \\\\
In Section \ref{sec:simulations} we shall consider the case where $r_{AC}=r_{CA}=0$, i.e. $r_{AB}=1-r_A$ and $r_{CB}=1-r_C$. For example, regimes $A$, $B$ and $C$ might describe working status of some technical system with many components involved. Then $A$ could correspond to the state where all the components work, $B$ could be the state where at least one of the components is broken and $C$ could be the state where all the components are broken and the system is not working. In this example it's natural to assume that all the components cannot break down at the same time and the broken components cannot be fixed within exactly the same time, thus $r_{AC}=r_{CA}=0$. Moreover, let us
assume that the observation space is $\X=\{1,2\}$ and the observations $X_1,X_2,X_3,\ldots$ evolve in regimes $A$, $B$ and $C$ as Markov chains with transition matrices
\begin{equation}
P_A=\left(
  \begin{array}{cc}
    1-\epsilon_1 & \epsilon_1 \\
    \epsilon_2 & 1-\epsilon_2 \\
  \end{array}
\right),\quad
   P_B= \left(
  \begin{array}{cc}
    {\epsilon_2\over \epsilon_1+\epsilon_2} & {\epsilon_1\over \epsilon_1+\epsilon_2}\\
    {\epsilon_2\over \epsilon_1+\epsilon_2} & {\epsilon_1\over \epsilon_1+\epsilon_2} \\
  \end{array}
\right),\quad P_C= \left(
  \begin{array}{cc}
   \delta_1 & 1-\delta_1 \\
    1-\delta_2 & \delta_2 \\
  \end{array}
\right).
\end{equation}
The parameters $\epsilon_1$, $\epsilon_2$, $\delta_1$, $\delta_2$ could be chosen so that they are
all rather small (less than $0.5$). Then regime $A$ corresponds to longer blocks (i.e. the process $X$ jumps less in regime $A$), the regime $B$ corresponds to HMM ({conditionally} independent observations) and regime $C$ corresponds to the case of shorter blocks (i.e. the process $X$ jumps more in regime $C$ compared to regime $A$). In order to stress that the difference between regimes is solely the dependence structure, one might choose
$\epsilon_1$, $\epsilon_2$, $\delta_1$, $\delta_2$ so that the proportion of ones and twos in every regime is equal. Then the following equality must hold:
\begin{equation} \label{paramcond}
{1-\delta_2\over 2-(\delta_1+\delta_2)}={\epsilon_2\over \epsilon_1+\epsilon_2}.\end{equation}
\paragraph{An example of HMM-DN by $X$.} It is important to realize that when the matrices $P_a, a\in \Y$, are quite similar to each other or
the matrix $P_Y$ has certain properties, then it might happen that the $X$-process is a homogeneous Markov chain. For example, let us consider a model with the following transition matrices:
\[P_Y=\left(
  \begin{array}{ccc}
    r & 1-r & 0 \\
    {(1-r)/2}& r & {(1-r)/2} \\
    0 & 1-r & r \\
  \end{array}
\right), \quad
 P_A=\left(
  \begin{array}{cc}
    1-\epsilon & \epsilon \\
    \epsilon & 1-\epsilon \\
  \end{array}
\right), \quad
  P_C= \left(
  \begin{array}{cc}
   \epsilon & 1-\epsilon \\
    1-\epsilon & \epsilon \\
  \end{array}
\right), \]
\[  P_B= \left(
  \begin{array}{cc}
   {1/2} & {1/2}\\
   {1/2} & {1/2}\\
  \end{array}
\right), \quad
 P_{AB}=\left(
  \begin{array}{cc}
   p_{AB} & 1-p_{AB} \\
    1-p_{AB} & p_{AB} \\
  \end{array}
\right), \quad
   P_{BA}= \left(
   \begin{array}{cc}
   p_{BA} & 1-p_{BA} \\
    1-p_{BA} & p_{BA} \\
  \end{array}
\right), \]
\[
P_{BC}= \left(
  \begin{array}{cc}
   p_{BC} & 1-p_{BC} \\
    1-p_{BC} & p_{BC} \\
  \end{array}
\right), \quad
P_{CB}= \left(
  \begin{array}{cc}
   p_{CB} & 1-p_{CB} \\
    1-p_{CB} & P_{CB} \\
  \end{array}
\right). \]
It is shown in \cite{avans} that such a model is HMM-DN by $X$ (hence $X$ is a Markov chain) if and only if the following condition holds:
\begin{equation}\label{oi}
{r\over 1-r}\left({1\over 2}-\epsilon\right)+p_{AB}={p_{BA}+p_{BC}\over 2}=p_{CB}-{r\over 1-r}\left({1\over 2}-\epsilon\right).
\end{equation}
Then $X$ is a homogeneous Markov chain with transition matrix
$$\left(
    \begin{array}{cc}
      \alpha & 1-\alpha \\
      1-\alpha & \alpha \\
    \end{array}
  \right),\quad \mbox{where} \quad \alpha=r({1}-\epsilon)+(1-r)p_{AB}.$$
In particular, when $2{r\over 1-r}({1\over 2}-\epsilon)<1$, then one can choose inter-regime matrices so that (\ref{oi}) holds. The condition $2{r\over 1-r}({1\over 2}-\epsilon)<1$ holds when $\epsilon\approx {1/2}$, so that $P_A\approx P_B\approx P_C$, or when $r\leq  {1/2}$
(recall that $\e<{1/2}$), making the holding times in regimes $A$ and $C$ relatively short. This example illustrates that for a meaningful PMM the matrices $P_a, a\in \Y$, should not be so similar to each other and the holding times in different regimes should not be very short, otherwise $X$ might turn out to be a homogeneous Markov chain and there is no need to model it with PMMs.
\subsection{Semi-Markov model}\label{sec:semi-markov}
Let $\X=\{A,B,C,\ldots\}$ stand for a possibly infinite alphabet. A semi-Markov process is a generalization of a Markov chain on $\X$ where the sojourn times (times the chain spends in a given state) are not necessarily geometrically distributed. Let the sojourn time distribution for every $a\in \X$ be given by a probability distribution $q_a$, $q_a(k)\geq 0$ for every $k=1,2,\ldots$. After the process has spent a random time with distribution $q_a$ in state $a$, it jumps to the other state $b\ne a$ with probability $p_{ab}=P(X_{t+1}=b|X_{t+1}\ne a,X_t=a)$. Obviously $p_{aa}=0$ for every $a\in \X$. These probabilities form the transition matrix $P=(p_{ab})$, which we shall call {\it the jump matrix of $Z=(X,Y)$}.
There are various ways for considering a semi-Markov chain as a PMM. A common way is to consider a semi-Markov model as an HMM $(Y,X)$, where $Y$ is a  Markov chain with transition matrix $P$ and the values of $X_t$ are the sojourn times of $Y_t$. Thus, the distributions $q_a$ correspond to emission distributions. Sometimes  (see e.g. \cite{Psemi}), it is useful to consider it as a PMM $Z=(X,Y)$, where the state space ${\cal Z}$ consists of pairs ${\cal Z}=\{(a,k): a\in \X, q_a(\geq k)>0\}$, where $q_a(\geq k)=\sum_{i\geq k}q_a(i)$.
We see that $Y$ takes values in $\mathbb{N}^+$ and ${\cal Z}$ can be a proper subset of $\X\times \mathbb{N}^+$. The possibly infinite transition matrix of $Z$ consists mostly of zeros and is given by $$P(Z_{t+1}=(b,l)|Z_t=(a,k))=\left\{
                          \begin{array}{ll}
                            1, & \hbox{when $b=a$, $l=k-1>1$;} \\
                            p_{ab}q_b(l), & \hbox{when $b\ne a$, $k=1$;}\\
                            0, & \hbox{else.}
                          \end{array}
                        \right.$$
Thus, when $Z_t=(a,j)$ with $j>1$, then the only possible transition is to $(a,j-1)$. When $j=1$, then $X_{t+1}$ cannot be in state $a$ any more. An example of a realization $z^{10}$ of such a PMM might for example be $$(B,2),(B,1),(C,4),(C,3),(C,2),(C,1),(A,1),(B,4),(B,3),(B,2),$$  and we can  see that up to the last block the values of $Y_t$ can be actually read from $X_t$. The obtained model $Z$ is an example of a PMM that is neither HMM-DN nor HMM-DN by $X$.  Clearly neither of the marginal processes is a Markov chain. When $\X$ is finite, the matrix $P$ has a unique stationary distribution $(\pi_a)$ and all the sojourn times have finite expectations, let them be denoted by $\mu_a<\infty$. Then $Z$ has a unique stationary distribution $\pi(a,k)$, where
$$\pi(a,k)={\pi_a q_a(\geq k)\over \sum_{b\in \X}\pi_b \mu_b}.$$
\subsection{Semi-Markov regime-switching model}\label{sec:semi-regime}
If the goal is to model an inhomogeneous Markov chain with the sojourn times not being geometrically distributed, then the two PMMs -- semi-Markov and regime-switching model -- could be merged into one PMM as follows. Let $Y$ be the semi-Markov PMM considered in the previous example. Thus the states of $Y$ are pairs $(a,k)$, where $a$ is the regime and $k$ indicates the time left to be in regime $a$. Let $X$ stand for observations, $X_t$ takes values in $\{1,2,\ldots\}$. As in the regime-switching model, there corresponds a $|\X|\times |\X|$ transition matrix $P_a$ to every regime $a$. In order to specify the model, one has to choose the inter-regime matrices $P_{ab}$ as well. Just as in the regime-switching model, the corresponding PMM $(X,Y)$ can be defined with the following transition matrix:
$$P(X_{t+1}=j,Y_{t+1}=(b,l)|X_{t}=i,Y_{t+1}=(a,k))=\left\{
                                                     \begin{array}{ll}
                                                       P_{a}(i,j), & \hbox{when $a=b$, $k=l+1$;} \\
                                                       p_{ab}P_{ab}(i,j)q_b(l), & \hbox{when $a\ne b$, $k=1$;} \\
                                                       0, & \hbox{else;} \\
                                                     \end{array}
                                                   \right.$$
where $q_a(\cdot)$ are the distributions of sojourn times and $(p_{ab})$ is the jump matrix of $Y$. The obtained PMM is an HMM-DN, so when $y^n$ is a realization of $Y^n$ such that $y_{t+1}=(a,k)$, where $k>1$ (implying that $y_t=(a,k+1)$), then it holds that
$$P(X_{t+1}=j|X_t=i,Y^n=y^n)=P(X_{t+1}=j|X_t=i,Y_t=(a,k+1),Y_{t+1}=(a,k))=P_a(i,j).$$
Finally, let us remark that since $Y$ itself is a PMM, say $(U,V)$, then the obtained PMM can be considered as a three-dimensional Markov chain $(X,U,V)$. Such models are known as {\it triplet Markov models (TMM)}, see \cite{gorynin18}. Every TMM can obviously be considered as a PMM by considering two of the three marginal processes as one, thus we can consider the following PMMs: $(X,(U,V))$, $((X,U),V)$ or $((X,V),U)$.
\section{Segmentation and risks}\label{sec:segmentation}
The term `hidden Markov model' reflects the situation where the realization of $X$-chain is observed, but the realization of the Markov chain $Y$ is not observed, hence it is hidden. We now have a more general model -- PMM $(X,Y)$ --, but we still assume that a realization $x^n$ of $X^n$ is observed, whilst the corresponding realization of $Y^n$ is unknown. Thus, $x^n$ can be considered as a sample or observations and $n$ is sample size.
The {\it segmentation problem} consists of estimating the
unobserved realization of the underlying process $Y^n$ given
observations $x^n$. Formally, we are looking for a mapping
$g:{\cal X}^n \to {\cal Y}^n$ called a {\it classifier} or {\it
decoder}, that maps every sequence of observations into a state
sequence. The best classifier $g$ is often defined via a {\it loss function} $L: {\cal Y}^n\times {\cal Y}^n \to [0,\infty],$ where
$L(y^n,s^n)$ measures the loss when the actual state sequence is $y^n$ and the estimated sequence is $s^n$. For
any state sequence $s^n\in {\cal Y}^n$, the expected loss for given $x^n$ is called {\it conditional risk}:
\[\label{risk}
R(s^n|x^n):=E[L(Y^n,s^n)|X^n=x^n]=\sum_{y^n\in
{\cal Y}^n}L(y^n,s^n)p(y^n|x^n).\]
The best classifier is defined as a state sequence minimizing the conditional risk:
$$g^*(x^n)=\arg\min_{s^n\in {\cal Y}^n}R(s^n|x^n).$$
For an overview of risk-based segmentation with HMMs, see
\cite{seg,intech,chris}. The two most common loss functions used in practice are the global loss function $L_{\infty}$,
\[\label{symm-n}
L_{\infty}(y^n,s^n):=\left\{
               \begin{array}{ll}
                 1, & \hbox{if $y^n\ne s^n$,} \\
                 0, & \hbox{if $y^n=s^n$,}
               \end{array}
             \right.
\]
and the loss function $L_1$ obtained with the {\it pointwise loss function} $l: \Y\times \Y \to [0,\infty)$, where $l(s,s)=0$ $\forall s \in \cal Y$,
\begin{equation}\label{kadu-p}
  L_1(y^n,s^n):={1\over n}\sum_{i=1}^n l(y_i,s_i).
\end{equation}
Observe that the loss function $L_{\infty}$ penalizes all differences equally: no matter
whether two sequences $a^n$ and $b^n$ differ at one entry or at all entries, the penalty is one. The loss function $L_1$
on the other hand penalizes differences entrywise. The conditional risk corresponding to $L_{\infty}$ and denoted by $R_{\infty}$ is
$R_{\infty}(s^n|x^n)=1-p(s^n|x^n)$, thus the best classifier $v$ maps every
sequence of observations into sequence $s^n$ with maximum posterior probability:
$$v(x^n):=\arg\max_{s^n\in \Y^n}p(s^n|x^n).$$
Any state path $v$  maximizing $p(s^n|x^n)$ is called the {\it Viterbi path} or {\it Viterbi alignment} (it might not be unique). The best classifier in the case of $L_1$ in (\ref{kadu-p}) is obtained pointwise:
$g^*=(g^*_1,\ldots,g^*_n)$, where
\begin{equation}\label{pointwise}
g^*_t(x^n)= \arg\min_{s \in {\cal Y}}E[l(Y_t,s)|X^n=x^n]= \arg\min_{s\in \Y }\sum_{y\in \Y}
l(y,s) p_t(y|x^n).\end{equation} If
$$
l(s,s')=\left\{
  \begin{array}{ll}
    0, & \hbox{if $s=s'$,} \\
    1, & \hbox{if $s\ne s'$,}
  \end{array}
\right.$$
then the loss function $L_1$  counts pointwise differences
between $y^n$ and $s^n$. Thus the corresponding conditional risk measures the {\it expected number of
classification errors} of $s^n$ given the observations $x^n$ and can be calculated as follows:
$$R_1(s^n|x^n):=1-{1\over n}\sum_{t=1}^np_t(s_t|x^n).$$
It follows that the best classifier under $R_1$ (let us denote it by $u$) minimizes the expected number of classification errors
and it can be calculated pointwise:
\begin{equation*}
u_t(x^n)=\arg\max_{y\in {\cal Y}}p_t(y|x^n),\quad t=1,\ldots,n.
\end{equation*}
We will call any such $u$ a {\it pointwise maximum aposteriori (PMAP)} path. In PMM literature often the name {\it maximum posterior mode (MPM)} is used, see e.g. \cite{P03,P04,Pzebra,P12,gorynin18}.
\subsection{Logarithmic and hybrid risks} Define the following logarithmic risks:
\begin{align*}
\R_{\infty}(s^n|x^n)&:=-{1\over n}\ln
p(s^n|x^n),\quad \R_1(s^n|x^n):=-{1\over
n}\sum_{t=1}^n\ln p_t(s_t|x^n),
\end{align*}
then the Viterbi path $v(x^n)$ minimizes
$\R_{\infty}(\cdot|x^n)$ and the PMAP path $u(x^n)$
minimizes
$\R_{1}(\cdot|x^n)$.\\\\
The Viterbi path has biggest posterior probability, but it might be
inaccurate when it comes to the number of pointwise errors. The PMAP path on the other hand is the most accurate state path in terms of expected number of errors, but it might have very low or even zero posterior probability. In what follows, paths with zero posterior probability are called {\it inadmissible}. Often the
goal is to find a state path that combines the two desired properties: it has a relatively big likelihood and relatively high accuracy.
In \cite{seg}, a family of {\it hybrid paths} was defined. A hybrid path operates between the PMAP and Viterbi path and is the solution to the following problem:
\begin{equation}\label{hybrid}
\min_{s^n}[\R_1(s^n|x^n)+C\R_{\infty}(s^n|x^n)]\quad
\Leftrightarrow\quad \max_{s^n}\big[\sum_{t=1}^n \ln
p_t(s_t|x^n)+C\ln p(s^n|x^n)],
\end{equation}
where $\R_C=\R_1(s^n|x^n)+C\R_{\infty}(s^n|x^n)$ is the hybrid risk and $C\geq 0$ is a regularization constant. The case $C=0$
corresponds to the PMAP path and it is easy to see that increasing
$C$ increases the posterior probability ($\R_{\infty}$-risk) and
decreases the accuracy ($\R_1$-risk) (see, e.g. Lemma 16 in \cite{seg}). If $C$ is sufficiently big (depending on the model and
$x^n$), then the solution is given by the Viterbi path. We
now give an interpretation of the hybrid risk in terms of blocks.\\\\
As a remedy against zero-probability PMAP paths, Rabiner \cite{rabiner} proposed in his seminal
tutorial the following: instead of maximizing the sum
$\sum_{t=1}^n p_t(s_t|x^n)$ over all $s^n\in {\cal Y}^n$, consider blocks of size $k$ and maximize
\begin{equation}\label{k-block}
p(s_1^k|x^n)+p(s_2^{k+1}|x^n)+\ldots + p(s_{n-k+1}^n|x^{n}).
\end{equation}
The case $k=1$ corresponds to the
PMAP path, the bigger $k$, the `closer' we come to the Viterbi
path. This idea can be generalized by defining a $k$-block
loss function as follows:
\begin{equation}\label{Lk}
l_k: \Y^k \times \Y^k \to [0,\infty),\quad
L_k(y^n,s^n):={1\over {n-k+1}}\sum_{t=0}^{n-k}
l_k(y_{t+1}^{t+k},s_{t+1}^{t+k}).
\end{equation}
For HMMs, the case $k=2$ is studied in \cite{chris} under the name
{\it Markov loss function}. When
\begin{align}\label{pll}
l_k(y_{t+1}^{t+k},s_{t+1}^{t+k})=I_{\{y_{t+1}^{t+k} \ne s_{t+1}^{t+k}\}},\end{align}
then minimizing the risk corresponding to the loss function
$L_k$ is equivalent to maximizing (\ref{k-block}). The case $k=2$ corresponds to the state path that
maximizes the expected number of correctly classified pairs (transitions). Unfortunately, the path minimizing the expected
$L_k$-loss can still have posterior probability 0 (see the example in \cite{seg}), therefore we use the following modification of
(\ref{k-block}). Define for any $k=1,2,\ldots,n$,
\begin{align*}
\R_k(s^n|x^n)&:=-{1\over n}\sum_{t=1-k}^{n-1}\ln p\left(s_{\max(t+1,1)}^{\min(t+k,n)}\big|x^n\right).
\end{align*}
For example, if $k=3$ and $n=7$, then denoting $r(s_{u}^t):=\ln p(s_{u}^t|x^n),$ we have
\begin{align*}
 -n\R_3(s^n|x^n)=r(s_1)+r(s_{1}^2)+r(s_{1}^3)+r(s_{2}^4)+r(s_{3}^5)+r(s_{4}^6)+r(s_{5}^7)+r(s_{6}^7)+
 r(s_7).
 \end{align*}
Thus, for small $k$, $-n \R_k(s^n|x^n)$ is basically the sum
$$
\ln p(s_{1}^k|x^n)+\ln p(s_{2}^{k+1}|x^n)+\cdots +\ln p(s_{n-k+1}^n|x^n).$$
Let $u_k(x^n)$ minimize $\R_{k}(\cdot|x^n)$. Clearly the PMAP
alignment $u(x^n)$ minimizes $\R_{1}(\cdot|x^n)$,  hence
$u_1(x^n)=u(x^n)$. The connection between hybrid risks and
blocks is given by the following proposition.
\begin{proposition}\label{prop} Let $1<k \leq n$, then for every $s^n\in {\cal Y}^n$,
\begin{equation}\label{rec}
\R_k(s^n|x^n)=\R_{\infty}(s^n|x^n)+\R_{k-1}(s^n|x^n).
\end{equation}
\end{proposition}
\begin{proof} For HMMs the proposition was proved in \cite{seg} (Theorem 6 and Corollary 7). Let us denote
\[ \bar{U}_k(s^n|x^n):= \prod_{t=1-k}^{n-1} p \left( s_{\max(t+1,1)}^{\min(t+k,n)} \Big| x^n \right), \]
then $\R_k(s^n|x^n)=-1/n \ln \bar{U}_k(s^n|x^n)$. In \cite{seg} it was shown by applying the Markov property that for any realization $s^n$
of the first order Markov chain, $\bar{U}_k(s^n)=p(s^n)\bar{U}_{k-1}(s^n)$. Since in the case of a PMM $(X,Y)$, $Y^n|x^n$ is a first order
(inhomogeneous) Markov chain, the proof immediately holds for PMMs. \end{proof}\\\\
From Proposition \ref{prop} it follows that
$$\R_k(s^n|x^n)=(k-1)\R_{\infty}(s^n|x^n)+\R_{1}(s^n|x^n),$$
thus when $C=k-1$, the hybrid risk $C\R_{\infty}+\R_1$ is actually
the $k$-block risk $\R_k$. Therefore, the hybrid risk can be considered
as a generalization of the $k$-block risk for non-integer value of $C$.
\subsection{Algorithms}\label{sec:algorithms}
The block risks and hybrid risks are meaningful and theoretically
justified, but the direct optimization of any risk over $\Y^n$ is
beyond computational capacity even for moderate $n$. Therefore, dynamic programming algorithms similar to the Viterbi one should be
applied. For HMMs the algorithm for the hybrid risk was worked out in \cite{seg}. In the present paper we state the dynamic programming
algorithm also for PMMs. Let us denote the states of $\cal Y$ by $1,\ldots,l$.
\begin{proposition} \label{prop1}
The state path(s) minimizing the hybrid risk
\begin{equation}\label{hybridB}
C\R_{\infty}(s^n|x^n)+B\R_1(s^n|x^n) = -{C \over n} \ln p(s^n|x^n) - {B \over n} \sum_{t=1}^n \ln p_t(s_t|x^n)
\end{equation}
can be found by the following recursion. Define the following scores:
\begin{align*}
\delta_1(j)&=C\ln p(x_1,j)+B\ln p_1(j|x^n), \quad \forall j \in \cal Y,\\
\delta_{t+1}(j)&=\max_{i \in \cal Y}\Big(\delta_t(i)+C\ln p(x_{t+1},j|x_t,i)\Big)+B\ln p_{t+1}(j|x^n),
\quad t=1,\ldots,n-1,\quad \forall j \in \cal Y.
\end{align*}
Using the scores $\delta_t(j)$,
define the backpointers $\psi_t(j)$ and the terminal state $\psi_n$ as follows:
\begin{align*}
\psi_{t}(j)&=\arg\max_{i \in \cal Y}\Big[\delta_t(i)+C\ln p(x_{t+1},j|x_t,i)\Big], \quad t=1,\ldots,n-1; \\
\psi_n & = \arg\max_{i \in \cal Y}\delta_n(i).
\end{align*}
The optimal state path $\hat{y}=(\hat{y}_1,\ldots,\hat{y}_n)$ minimizing the hybrid risk in (\ref{hybridB}) can be obtained as
\[ \hat{y}_n= \psi_n, \quad \hat{y}_t=\psi_t(\hat{y}_{t+1}), \quad t=n-1,\ldots,1.\]
\end{proposition}
\begin{proof} The proof of the proposition can be performed using induction. Observe that minimizing (\ref{hybridB}) over all $s^n$
is equivalent to
\[ \max_{s^n} [ C \ln p(s^n,x^n)+ B\sum_{t=1}^n \ln p_t(s_t|x^n)]  =\max_{s^n} [ C\ln p(s_1,x_1)+C \sum_{t=2}^n \ln p(s_t,x_t|s_{t-1},x_{t-1})+ B\sum_{t=1}^n \ln p_t(s_t|x^n) ].   \]
Let
\begin{align*}
U(s_1)&=C\ln p(s_1,x_1)+ B \ln p_1(s_1|x^n), \\
U(s^t)&=U(s_1)+ \sum_{u=2}^t \left[C \ln p(x_u,s_u|x_{u-1},s_{u-1})+ B\ln p_u(s_u|x^n)\right], \quad t=2,\ldots,n.
\end{align*}
Then
\[ U(s^{t})=U(s^{t-1}) +C\ln p(x_{t},s_{t}|x_{t-1},s_{t-1})+B\ln p_{t}(s_{t}|x^n), \quad t=2,\ldots,n, \]
and we can see directly that $\delta_t(j)$ gives the score of the state path $s^{t}$ that minimizes the hybrid risk
and ends in state $j$, that is \[ \delta_t(j)=\max_{s^{t}: \, s_t=j}  U(s^{t}).\]
By induction on $t$ this holds also for $\delta_n(j)$, therefore backtracking from $\psi_n$ gives us the optimal hybrid state path.
\end{proof}\\\\
Observe that the constant $B$ in (\ref{hybridB}) is redundant, since in
practice one can always take $B=1$ and vary the constant $C$, just
like in (\ref{hybrid}). The reason for adding $B$ to the recursion is that it immediately allows
to obtain the Viterbi algorithm by taking $B=0$ and $C=1$.\\\\
In the special case of HMM-DN the recursion is
$$\delta_{t+1}(y_{t+1})=\max_{y_t}\Big(\delta_t(y_t)+C\ln p(y_{t+1}|y_t)+C\ln p(x_{t+1}|y_{t+1}, y_t,x_t)\Big)+B\ln p_{t+1}(y_{t+1}|x^n)$$
and for a Markov switching model $p(x_{t+1}|y_{t+1}, y_t,x_t)=p(x_{t+1}|y_{t+1},x_t)$. The algorithm above involves applying {\it forward-backward algorithms} to find the probabilities
$p_t(y|x^n)$ for every $y\in \Y$ and $t$. The  forward algorithm finds recursively the probabilities  $p(x^t,y_t)$:
$$p(x^{t},y_{t})=\sum_{y_{t-1}\in \Y}p(x_{t},y_{t}|x_{t-1},y_{t-1})p(x^{t-1},y_{t-1}),$$
and the backward algorithm finds recursively the probabilities $p(x_{t+1}^n|x_t,y_t)$ as follows:
$$p(x_{t+1}^n|x_{t},y_{t})=\sum_{y_{t+1}\in \Y}p(x_{t+1},y_{t+1}|x_{t},y_{t})p(x_{t+2}^n|x_{t+1},y_{t+1}).$$
In practice the scaled versions of these probabilities are used, see \cite{P04}. The (scaled) forward-backward algorithms work essentially in the same way as for HMMs, this is all due to the Markov property.
Observe that when the model is HMM-DN by $X$ (like in our first example in Subsection \ref{sec:milano}), then by (\ref{alfa}),  $p(y_t|x^n)=p(y_t|x^t)$, thus $p(y_t|x^n)$ can be obtained by the forward recursion only.
The scaled forward recursion in this particular case is simply
\begin{equation}\label{scaledMilano}
p(y_t|x^t)=\sum_{y_{t-1}}p(y_t|x_t,x_{t-1},y_{t-1})p(y_{t-1}|x^{t-1}).
\end{equation}
For large $n$, replacing the forward-backward recursion by the forward one might be
a big computational advantage. Moreover, when the model is HMM-DN by $X$ and the time-reversed model is HMM-DN by $X$ as well, then for the stationary chain it holds that $p(y_t|x^n)=p(y_t|x_t)$. Therefore, in this case the probabilities $p(y_t|x^n)$ can be found without any forward-backward algorithms, which makes these models especially appealing from the computational point of view.
\paragraph{Rabiner $k$-block algorithm.}
In practice it is interesting to compare the state path estimates of the hybrid approach to the Rabiner
$k$-block state path estimates defined in (\ref{k-block}). Next we will give the algorithm for computing
the Rabiner $k$-block state path estimates.
\begin{proposition}\label{prop2}
The state path(s) minimizing the risk function corresponding to the Rabiner $k$-block approach and thus, maximizing the sum of probabilities in  (\ref{k-block}), can be found by the following recursion. Define for every $a\in {\cal Y}^{k-1}$ scores $\delta_t(a)$ and backpointers $\psi_t(a)$ as follows:
\begin{align*}
\delta_1(a)&=  \max_{y_1 \in \cal Y} p (y_1,y_2^k=a|x^n), \\
\psi_1(a)& = \arg\max_{y_1 \in \cal Y} p (y_1,y_2^k=a|x^n),  \\
\delta_t(a)&= \max_{y_t \in \cal Y}\left(\delta_{t-1}(y_t,a_1,\ldots,a_{k-2})+p(y_t,y_{t+1}^{t+k-1}=a|x^n)\right), \quad t=2,\ldots,n-k+1, \\
\psi_t(a)&= \arg\max_{y_t \in \cal Y}\left(\delta_{t-1}(y_t,a_1,\ldots,a_{k-2})+p(y_t,y_{t+1}^{t+k-1}=a|x^n)\right), \quad t=2,\ldots,n-k+1.
\end{align*}
Let $a_{end}=\arg\max_{a\in {\cal Y}^{k-1}} \delta_{n-k+1}(a)$. Then the state path $\hat{y}=(\hat{y}_1,\ldots,\hat{y}_n)$ maximizing (\ref{k-block}) can be obtained as
\[ \hat{y}_{n-k+2}^n=a_{end}, \quad \hat{y}_t=\psi_t(\hat{y}_{t+1}^{t+k-1}),\quad t=n-k+1,\ldots,1. \]
\end{proposition}
{\begin{proof} The proof is analogous to the proof of Proposition \ref{prop1}. \end{proof}}
%
%
\paragraph{Remark.} When $k=2$, the scores have to be calculated just for every $j\in \cal Y$, then
\[ \delta_t(j)= \max_{y_{t} \in \cal Y}\left(\delta_{t-1}(y_{t})+p(y_{t},y_{t+1}=j|x^n)\right), \quad t=2,\ldots,n-k+1. \]
\section{Behaviour of different state path estimators}\label{sec:simulations}
Let us now consider the hybrid risk with $B=1$, i.e. the optimization problem in (\ref{hybrid}). We know that the solution of (\ref{hybrid}) for $C=0$ corresponds to the PMAP path and the solution for large $C$ corresponds to the Viterbi path.
Let $C_o$ be the smallest constant such that the solution of (\ref{hybrid}) is a Viterbi path for every $C>C_o$. When $\Y$ is finite, then also the set $\Y^n$ is finite, thus $C_o$ surely exists. When $C=C_o$, then
(\ref{hybrid}) has at least two solutions: one of them is the Viterbi path and one of them is not; when  $C<C_o$, then none of the hybrid paths is Viterbi. It is also important to observe that in case the Viterbi path is not unique, it is sometimes
meaningful to optimize (\ref{hybrid}) with $C>C_o$ rather than to run the Viterbi algorithm with some tie-breaking rule. Because although different Viterbi paths have the same $\bar{R}_{\infty}$-risk, they might have different
 $\bar{R}_{1}$-risks, therefore the solution of (\ref{hybrid}) corresponds to the Viterbi path $v$ that maximizes $\sum_t p(v_t|x^n)$ and has therefore the minimum expected number of classification errors amongst all the Viterbi paths -- {\it primus inter pares}.
\\\\
However, when the goal is to find a hybrid path that is neither the Viterbi nor PMAP path, then only the range $(0,C_o)$ for $C$ is of interest. Obviously that range depends on the model, but as we shall see, it might very much depend also on the observation sequence $x^n$ and this dependence can be very unstable even for simplest models. More precisely, we consider the  model in Subsection \ref{sec:milano} and show that for every $M<\infty$, there exists $n$ and observations $x^n$, so that with $x^{n+3}=(x_1,\ldots, x_n,1,1,1)$ it holds that $C_o(x^{n+1})-C_o(x^n)>M$. In other words, adding three more observations increases $C_o$ tremendously.
\subsection{Variation of regularization constant $C$}\label{sec:example1} Consider the PMM defined in (\ref{milano}) with parameters $p'=p=6/7$ and $q'=q=8/35$, and $\lambda_1=0.9$, $\lambda_2=0.2$, $\mu_1=0.8$, $\mu_2=0.1$.
Then $\theta_1=0.6$, $\theta_2=0.4$ (see the appendix). Choose the initial distribution $\pi=(\pi(1,a),\pi(1,b),\pi(2,a),\pi(2,b))$ so that
\[ \pi(a|1)=P(X_1=1,Y_1=a)/P(X_1=1)=0.452328159645=:c_a. \]
1) At first consider the observation sequence
\[ x^n=(1,\underbrace{1,2,1}_1,\underbrace{1,2,1}_2,\ldots,\underbrace{1,2,1}_m), \]
thus the sequence has a particular pattern and is of length $n=3m+1$, $m=1,2,\ldots$. Let us denote $\alpha_t(y)=p(y_t=y|x^t)$, then by recursion (\ref{scaledMilano}),
\begin{align*}
  \alpha_1(a)& = c_a,\quad  \alpha_2(a) = \alpha_1(a)\lambda_1  + \alpha_1(b) \lambda_2 =0.5166297...,\quad  \alpha_3(a)= \alpha_2(a)\theta_1  +\alpha_2(b) \theta_2 =0.5033259... \\
  \alpha_4(a)& = \alpha_3(a) \mu_1 +\alpha_3(b) \mu_2 =c_a,\quad \alpha_5(a)=\alpha_2(a),\quad  \alpha_6(a)=\alpha_3(a),\quad \ldots
\end{align*} so that
$\alpha_t(y)=\alpha_{t-3}(y)$, $y\in \{a,b\}$, $t=4,\ldots,n$. For this observation sequence $x^n$, the PMAP path and Viterbi path are given by
\[
 \mbox{PMAP}: \quad b\underbrace{aab}_1 \underbrace{aab}_2 \ldots \underbrace{aab}_m, \quad \quad
 \mbox{Viterbi}: \quad bbbbbbb \ldots bbb. \\
\]
For any path $y^n$, thus the {\it C-score} with $B=1$ (recall (\ref{hybridB})) reads as follows:
$$C\ln p(y^n|x^n)+\sum_{t=1}^n\ln p_t(y_t|x^n),$$
and the hybrid path corresponding to $C$ maximizes the $C$-score. In our example, the last hybrid path before Viterbi (when increasing $C$)
is given by
\begin{equation} \label{borderpath}y^n=(b,\underbrace{a,b,b}_1, \underbrace{a,b,b}_2, \ldots \underbrace{a,b,b}_m) .\end{equation}
The difference between the $C$-scores of the Viterbi path and $y^n$ in $(\ref{borderpath})$ is given by
\[ C\ln p(bbbb\ldots bbb|x^n)-C\ln p(babb\ldots abb|x^n)+\sum_{t=1}^n \ln \alpha_t(b)- \sum_{t=1}^n \ln \alpha_t(y_t)  \]
\[ =C \ln { [(1-\lambda_2)(1-\theta_2)]^m \over [\lambda_2(1-\theta_1)]^m } + m \ln{\alpha_2(b)\over \alpha_2(a)} = Cm\ln 6 + m \ln{\alpha_2(b)\over \alpha_2(a)}.\]
Thus, if
\[ C>{ \ln(\alpha_2(a)/\alpha_2(b)) \over \ln 6}=:C_o, \]
then the hybrid path becomes Viterbi, and this holds independently of $m$ or sequence length $n$. Observe that $C_o=0.0371$ is very small in this example and in terms of block size the hybrid path corresponding to block length $k=2$ would here be equal to the Viterbi path.
\\\\
2) Add now the piece $111$ to the end of $x^n$ and consider
\[ x^n=(1,\underbrace{1,2,1}_1,\underbrace{1,2,1}_2,\ldots,\underbrace{1,2,1}_m,1,1,1). \]
The PMAP path and Viterbi path are now given by
\[
 \mbox{PMAP}: \quad b\underbrace{aab}_1 \underbrace{aab}_2 \ldots \underbrace{aab}_m aaa, \quad \quad
 \mbox{Viterbi}: \quad aaaaaaa \ldots aaaaaa. \\
\]
Let us compare the $C$-scores of the Viterbi path and constant path $bbbb\ldots b$:
\begin{align*}
 &C\ln p(b\ldots b|x^n)+m\sum_{t=1}^3 \ln \alpha_t(b)+\sum_{t=n-3}^n \ln \alpha_t(b)-C\ln p(a\ldots a|x^n)-m\sum_{t=1}^3 \ln \alpha_t(a)-\sum_{t=n-3}^n \ln \alpha_t(a)\\
&=C\ln {p(b\ldots b|x^n)\over  p(a\ldots a|x^n)}+m\sum_{t=1}^3 \ln {\alpha_t(b)\over \alpha_t(a)}+\sum_{t=n-3}^n\ln  {\alpha_t(b)\over \alpha_t(a)}=C\left(\ln{1-c_a\over c_a}+3 \ln{8\over 9}\right)+mD+E,
\end{align*}
where
\begin{align*}
&D:=\sum_{t=1}^3\ln {\alpha_t(b)\over \alpha_t(a)}, \,\quad E:=\sum_{t=n-3}^n\ln{\alpha_t(b)\over \alpha_t(a)}= \ln{1-c_a\over c_a} + \ln{\alpha_{n-2}(b)\over \alpha_{n-2}(a)} + \ln{\alpha_{n-1}(b)\over \alpha_{n-1}(a)}+
 \ln{\alpha_{n}(b)\over \alpha_{n}(a)}.\end{align*}
Observe that $D>0$ and $\ln((1-c_a)/c_a)+3 \ln(8/9)<0$, thus when
$$C<{mD+E\over \ln(c_a/(1-c_a))+3\ln(9/8)}=:C(m),$$
it holds that the Viterbi path is not the hybrid one. Hence
$C_o\geq C(m)$ and since $C(m)$ increases with $m$, so does $C_o$. The computations show that actually $C(m)=C_o$ when $m>8$.\\\\
We would like to emphasize that the observed instability of $C_o$ is due to the specific structure of $x^n$. When generating observation sequences randomly from the same model we can see that typically $C_o(x^n)<1$ up to $n=10000$ (slightly increasing with $n$); for $n=100000$, $C_o(x^n)$ might occasionally exceed 1 and reach up to 3.
\subsection{Dissimilarities of different hybrid paths}\label{sec:example2} In this example we continue to study the  model (\ref{milano}), but now with the following parameters: $p'=p=0.55$, $q'=q=0.8$, $\lambda_1=0.52$, $\lambda_2=0.8$, $\mu_1=0.6$, $\mu_2=0.9$.
To study the behaviour of the random variable $C_o(X^n)$ as well as hybrid paths, we generated 10 realizations $x^n$ of $X^n$ for $n=100,\, 1000,\,10000,\,100000$. For each observation sequence ($4\times 10=40$ sequences) we performed segmentation with PMAP and Viterbi, and we estimated $C_o$. For most cases
$C_o(x^n)\in [12.53,12.55]$, thus the hybrid path corresponding to the block length $k=14$ equals the Viterbi path (for most of the cases). Observe the difference with the previous example, where the same model with another parameter values gave typically much smaller $C_o$. The larger $C_o$ implies that in the present example there is more `space' between the PMAP and Viterbi path and this is due to the very weak dependence  between $X$ and $Y$. The general pattern here is that $C_o(x^n)$ is independent of sequence length; however we also observed that for two studied observation sequences (one of length 10000 and one of length 100000), $C_o(x^n)\in[33.21,33.22]$.  It seems that this behaviour depends on some particular subsequences or pieces of $x^n$ and removing that particular piece of observations would result in $C_o\approx 12.5$.
\\\\
To compare the state path estimates obtained with the PMAP, Viterbi, hybrid and Rabiner $k$-block algorithms, we studied closer the path estimates for the 10 observation sequences of length 100. Recall that in our model both marginal chains have the same transition matrix and the average block length of ones and twos (or $a$-s and $b$-s) is 2.2 and  1.25, respectively. The stationary distribution of $X$ and $Y$ is given by
$\pi=(0.8/(0.45+0.8),0.45/(0.45+0.8))$, thus there are almost twice as many ones expected in our observation sequences as twos.
In Figure \ref{StatePaths}, all the estimated state paths for one observation sequence are presented (in the order from top to down: true underlying state path, PMAP, hybrid block paths for $k=2,\ldots,14$, Rabiner block paths for $k=2,\ldots,14$, Viterbi). For better visibility, we have plotted the first 80 states of the path estimates.
\begin{figure}[ht!]
\centering
\includegraphics[height=10cm, width=17cm]{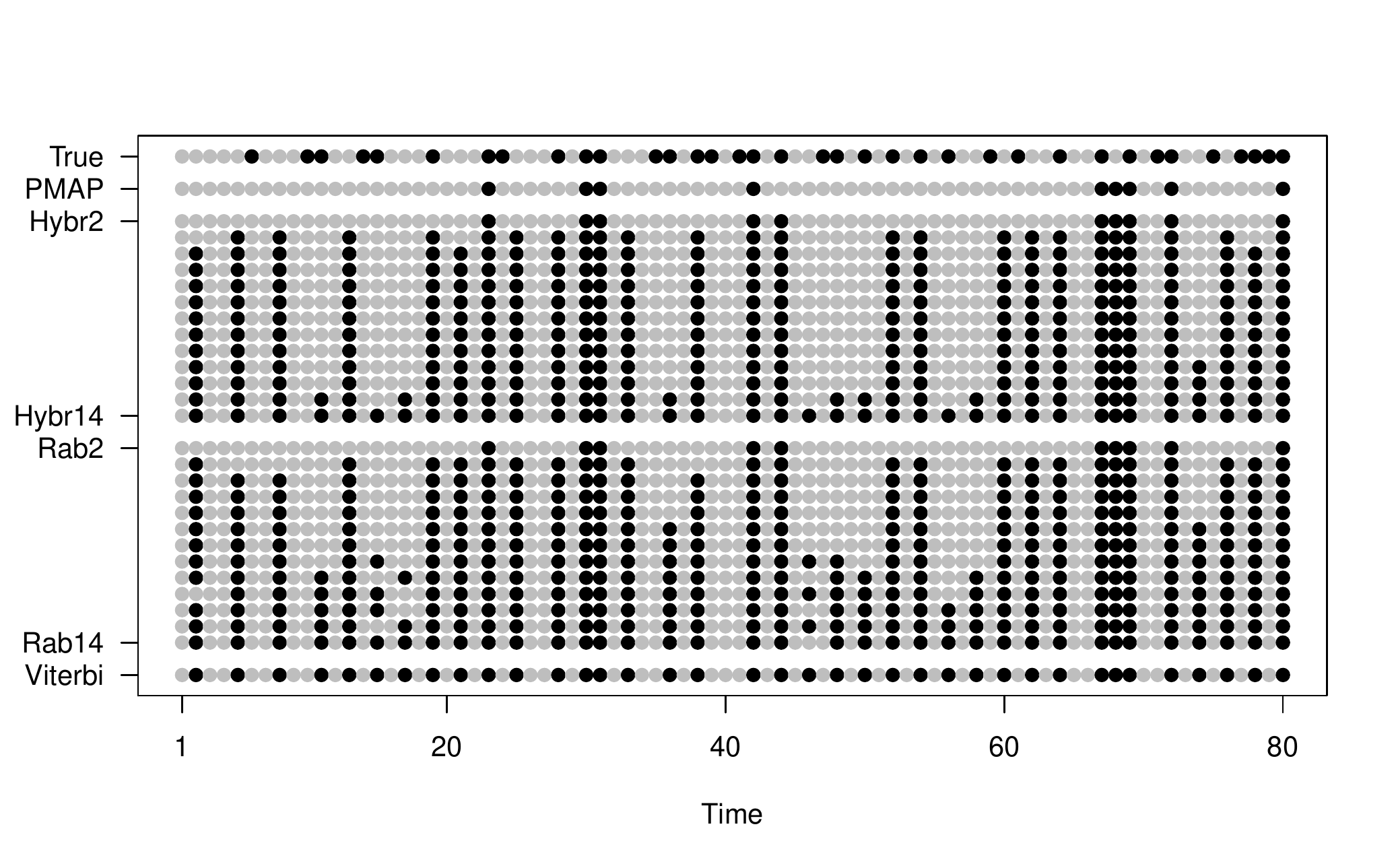}
\vspace{0.1cm}
\caption{\label{StatePaths} Studied state path estimates for one observation sequence in Example \ref{sec:example2}. From top to down: true underlying state path, PMAP, hybrid block paths for $k=2,\ldots,14$, Rabiner block paths for $k=2,\ldots,14$, Viterbi.}
\end{figure}
\noindent
We can see how the pattern changes when we move from PMAP to Viterbi: the number of dominating state $a$ (grey) decreases and the number of state $b$ (black) increases. When we study different block lengths to see how information from different neighbourhoods is accounted for, we can see that a larger change compared to PMAP occurs for $k=3$. It's also interesting to observe that the hybrid block estimates are the same for $k=4,\ldots,10$ (for all the ten observation sequences), thus increasing the block size doesn't change the path estimate for those $k$-values.
\\\\
To get a better overview of the behaviour of the estimated state paths, we present some summary statistics over 100 observation sequences of length $n=100$. For all the observation sequences we estimated the PMAP path, Rabiner and hybrid block paths for $k=2,\ldots,14$, and the Viterbi path. In Table \ref{SummaryHybr} the averages over 100 sequences of classification errors are presented for the PMAP path, hybrid block paths ($k=2,\ldots,14$) and Viterbi path.
%
\begin{table}[htbp!]
\begin{center}
{
\begin{tabular}{| c || c | c | c | }
  \hline
Path     &  Type I  & Type II  &  Errors   \\
  \hline
 PMAP &  2.97 (1.56) & 31.39 (3.91) & 34.36 (4.16) \\
 $k=2$ & 4.11 (2.00) & 30.67 (4.07) & 34.78 (4.22) \\
 $k=3$ & 14.66 (3.08) & 23.88 (3.61) & 38.54 (5.01) \\
 $k=4$ & 16.95 (3.23) & 22.88 (3.69) & 39.83 (4.81) \\
 $k=5$ & 16.95 (3.23) & 22.88 (3.69) & 39.83 (4.81) \\
 $k=6$ & 16.95 (3.23) & 22.88 (3.69) & 39.83 (4.81) \\
 $k=7$ & 16.95 (3.23) & 22.88 (3.69) & 39.83 (4.81) \\
 $k=8$ & 16.95 (3.23) & 22.88 (3.69) & 39.83 (4.81) \\
 $k=9$ & 16.95 (3.23) & 22.88 (3.69) & 39.83 (4.81) \\
 $k=10$ & 16.95 (3.23) & 22.88 (3.69) & 39.83 (4.81) \\
 $k=11$ & 18.12 (3.30) & 22.09 (3.63) & 40.21 (4.98) \\
 $k=12$ & 18.17 (3.28) & 22.07 (3.61) & 40.24 (4.95) \\
 $k=13$ & 23.20 (3.40) & 19.28 (3.40) & 42.48 (5.35) \\
 $k=14$ & 26.03 (3.64) & 17.78 (3.18) & 43.81 (5.58) \\
 Viterbi & 26.05 (3.64) & 17.77 (3.17) & 43.82 (5.58) \\
\hline
\end{tabular}}
\end{center}
\caption{\label{SummaryHybr} \footnotesize Averages of type I errors, type II errors and total number of classification errors over 100 sequences in Example \ref{sec:example2} for the PMAP path, hybrid block paths ($k=2,\ldots,14$) and Viterbi path. In the brackets the corresponding standard deviations are given.}
\end{table}
In this example we can make two types of classification errors: classify $a$ as $b$ (call it type I error) or classify $b$ as $a$ (type II error). To demonstrate further the difference between the PMAP and Viterbi paths, we present also averages of these classification errors separately. As the theory predicts, the number of classification errors increases with $k$. However, there is also a clear dependence between $k$ and error types: when $k$ is small ($k=1,2$) then the number of type I errors for this model is small. When $k$ increases and we move towards Viterbi, then the number of type I errors starts to increase and the number of type II errors decreases.
Notice also that when we compare the average number of pointwise errors for PMAP and Viterbi, then PMAP is about 10$\%$ better when we consider the total number of errors. The major difference between the results of the two algorithms is what type of errors we make.\\\\
{In Table \ref{SummaryRab} the same summary statistics are presented for the Rabiner $k$-block paths. The general behaviour concerning type I and type II errors is similar for the Rabiner $k$-block and hybrid  paths with $C=k-1$. The major difference is that the Rabiner algorithm gives more varying path estimates for $k=4,\ldots,10$, which is reflected in a smoother increase/decrease of the averages of type I/type II errors. In column \textit{Difference} of Table \ref{SummaryRab} the average number of pointwise differences (and its standard deviation) between the hybrid paths and Rabiner block paths is presented. For $k=12$ the average pointwise difference is 8.91 showing that the Rabiner and hybrid block path estimates can be pretty different (recall that $n=100$).}
\begin{table}[htbp!]
\begin{center}
\begin{tabular}{| c || c | c | c | c | }
  \hline
Path     &  Type I  & Type II  &  Errors & Difference  \\
  \hline
 PMAP &  2.97 (1.56) & 31.39 (3.91) & 34.36 (4.16) & na  \\
 $k=2$ & 4.11 (2.00) & 30.67 (4.07) & 34.78 (4.22) & 0  \\
 $k=3$ & 12.69 (3.26) & 25.22 (3.73) & 37.91 (4.53) & 5.11 (2.25) \\
 $k=4$ & 16.86 (3.23) & 23.00 (3.67) & 39.86 (4.92) & 1.09 (1.78) \\
 $k=5$ & 17.22 (3.27) & 22.69 (3.63) & 39.91 (4.81) & 0.50 (0.66) \\
 $k=6$ & 17.36 (3.32) & 22.64 (3.59) & 40.00 (4.84) & 0.95 (1.36) \\
 $k=7$ & 18.12 (3.29) & 22.26 (3.70) & 40.38 (4.93) & 1.93 (1.51) \\
 $k=8$ & 18.79 (3.21) & 21.66 (3.57) & 40.45 (4.88) & 3.34 (1.96) \\
 $k=9$ & 20.18 (3.35) & 20.95 (3.54) & 41.13 (4.97) & 5.30 (2.41) \\
 $k=10$ & 21.46 (3.20) & 20.28 (3.50) & 41.74 (4.98) & 7.47 (2.86) \\
 $k=11$ & 22.56 (3.42) & 19.51 (3.56) & 42.07 (5.36) & 7.78 (2.54) \\
 $k=12$ & 23.50 (3.47) & 19.15 (3.35) & 42.65 (5.35) & 8.91 (2.66) \\
 $k=13$ & 24.00 (3.39) & 18.86 (3.28) & 42.86 (5.33) & 4.52 (2.12) \\
 $k=14$ & 24.48 (3.58) & 18.60 (3.18) & 43.08 (5.47) & 2.59 (1.74) \\
 Viterbi & 26.05 (3.64) & 17.77 (3.17) & 43.82 (5.58) & na \\
\hline
\end{tabular}
\end{center}
\caption{\label{SummaryRab} \footnotesize Averages of type I errors, type II errors and total number of classification errors over 100 sequences in Example \ref{sec:example2} for the PMAP path, Rabiner $k$-block paths ($k=2,\ldots,14$) and Viterbi path. In the brackets the corresponding standard deviations are given. In column \textit{Difference} the average number of pointwise differences (and its standard deviation) between the hybrid block paths and Rabiner block paths is given.}
\end{table}
\subsection{Regime switching model and inadmissible state paths}\label{sec:example3} The main purpose of this example is to demonstrate possible inadmissibility of PMAP paths and that PMAP and Viterbi can give quite similar results in terms of classification errors. Consider a regime switching model with the following parameters:
$$ P_Y = \left(
  \begin{array}{ccc}
    0.95 & 0.05 & 0 \\
    (1-r_B)/2 & r_B & (1-r_B)/2 \\
    0 & 0.01 & 0.99 \\
  \end{array}
\right), \quad \mbox{where} \quad r_B=0.2,0.4,0.6,0.8;
$$
let $\epsilon_1=0.2$, $\epsilon_2=0.3$, $\delta_1=0.4$, $\delta_2=0.1$, thus
\[
P_A=\left(
  \begin{array}{cc}
   0.8  & 0.2 \\
   0.3  & 0.7 \\
  \end{array}
\right),\quad
   P_B= \left(
  \begin{array}{cc}
   0.6  & 0.4 \\
   0.6  & 0.4 \\
  \end{array}
\right),\quad P_C= \left(
  \begin{array}{cc}
   0.4 & 0.6 \\
   0.9 & 0.1 \\
  \end{array}
\right).
\]
Thus, we consider four different values of $r_B$ keeping the rest of the parameters fixed, and study how this affects segmentation results using different state path estimators. For this model (\ref{paramcond}) holds and the proportion of ones and twos in all the regimes is 0.6 and 0.4, respectively. Observe that the expected number of times the underlying chain is in regime B is according to the stationary distribution for cases  $r_B=0.2,0.4,0.6,0.8$ given by 6\%, 8\%, 11\% and 20\%, respectively. For given $r_B$, we generated 100 sequence pairs $(x,y)$ from the corresponding PMM with sequence length $n=1000$, and studied different state path estimates for those sequences. The results of the experiment are summarized in Table \ref{tableRb}. In this example regimes $A$ and $C$ are dominating and regime $B$ occurs, especially for $r_B=0.2$ and $r_B=0.4$, very rarely. Since the block lengths of ones and twos in regime $A$ are longer on average compared to regime $C$, it's quite easy to separate the two regimes based on observations. This means that for smaller  $r_B$ classification should be easier and the simulations confirm it -- we see that the average number of pointwise errors in the case $r_B=0.2$ is 11\% and 12\% for PMAP and Viterbi, for $r_B=0.4$ the corresponding numbers are 13\% and 14\%. When the frequency of regime $B$ increases with increasing $r_B$, the pointwise error rates also increase. For $r_B=0.8$ the error rates of PMAP and Viterbi are 22\% and 25\%. We can also see that the average number of pointwise differences between the PMAP and Viterbi path for $r_B=0.2,0.4,0.6$ is quite small: 40, 41 and 48, respectively. Thus, one could think that the PMAP and Viterbi path estimates are quite similar but this is not the case. The problem with PMAP paths for this model (with $r_B=0.2,0.4,0.6$) is that the path estimates are inadmissible because of the impossible transitions $1\to 3$ and $3 \to 1$. The inadmissibility of PMAP paths is also evident from the low frequencies of regime $B$ in the first 3 rows of column \textit{PMAP}. The average number of inadmissible transitions in the PMAP paths for each $r_B$ is given in column \textit{Inadm(PMAP)}. To exemplify inadmissibility of Rabiner $k$-block paths, the average number of inadmissible transitions is presented also for the Rabiner block paths with $k=2$ and $k=5$. We have also counted the number of admissible PMAP, Rabiner $2$- and $5$-block paths (if any), those numbers are presented in the brackets after the average number of inadmissible transitions. Thus, we can see that for $r_B=0.2$ and $r_B=0.4$, 8 and 4 Rabiner $5$-block paths were admissible, respectively. For $r_B=0.8$, only one PMAP path was inadmissible and there were 6 admissible Rabiner $2$-block paths. The fact that even Rabiner $5$-block paths might be inadmissible is alarming -- the intuition suggests that the longer the blocks, the closer the path is to the Viterbi path, but even the blocks of length 5 cannot guarantee admissibility of Rabiner paths in this example.\\\\
To conclude: since PMAP paths are inadmissible, in this example with $r_B=0.2,0.4,0.6$, one should use a hybrid path or Viterbi path as a hidden path estimate. When the purpose is to minimize the expected number of pointwise errors, the 2-block hybrid path could be used (the average number of pointwise errors is given in column \textit{Err(Hybr2)}) {or} any hybrid path with $C\in(0,1]$ and $B=1$ in (\ref{hybridB}).
%
%
%
\begin{table}[htbp!]
\begin{center}
\begin{tabular}{|c||c|c|c|c|c|c|}
  \hline
$r_B$ & PMAP  & Viterbi  & Err(PMAP) & Err(Vit) & Err(Hybr2) & Diff(PMAP/Vit)\\
  \hline
0.2 &  506/2/493 & 498/15/488 & 109 (20) &  119 (27) & 114 (22) & 40 (17) \\
0.4 &  495/5/500 & 491/15/495 & 127 (21) &  138 (26) & 130 (23) & 41 (15) \\
0.6 &  483/14/503 & 481/16/503 & 155 (29) & 168 (34) & 155 (30) & 48 (18) \\
0.8 &  463/111/426 & 507/34/460 & 225 (39) &  252 (48) & 228 (42) &  114 (34) \\
\hline
\end{tabular}
\vskip1\baselineskip\noindent
\begin{tabular}{|c||c|c|c|}
  \hline
$r_B$ & Inadm(PMAP) & Inadm(Rab2) & Inadm(Rab5) \\
  \hline
0.2 &   18.27 & 11.66 & 2.79 (8) \\
0.4 &   15.41 & 12.22 & 3.83 (4) \\
0.6 &    9.20  & 11.60 & 5.08  \\
0.8 &  0.01 (99)  &   3.48 (6) &  5.64 \\
\hline
\end{tabular}
\end{center}
\caption{\label{tableRb} \footnotesize Different summary statistics for 100 generated sequence pairs $(x^n,y^n)$ of length $n=1000$ from the model in Example \ref{sec:example3} with $r_B=0.2,0.4,0.6,0.8$. In columns \textit{PMAP} and \textit{Viterbi}, average frequencies (over 100 sequences) of regimes $A$, $B$ and $C$ are given for the PMAP and Viterbi path. Columns \textit{Err(PMAP)}, \textit{Err(Vit)} and \textit{Err(Hybr2)} give the average number of pointwise errors (over 100 sequences) of the PMAP, Viterbi and hybrid 2-block path compared to the true path, the corresponding standard deviations are presented in the brackets. Column \textit{Diff(PMAP/Vit)} reports the average number of pointwise differences between the PMAP and Viterbi path for the models with different $r_B$, the corresponding standard deviations are given in the brackets. Columns \textit{Inadm(PMAP)}, \textit{Inadm(Rab2)} and \textit{Inadm(Rab5)} report the average number of inadmissible transitions in the PMAP, Rabiner 2-block and Rabiner 5-block paths; the number of admissible paths if any is reported in the brackets.}
\end{table}
\\\\
We also studied the distribution of $C_o(X^n)$ for different  $r_B$. For each simulated observation sequence we calculated the smallest integer $k_o$ such that for $C\geq k_o$, the hybrid path equals the Viterbi path. Thus $k_o=\lceil C_o \rceil$. Recall that for $k \leq n$ the hybrid path with $C=k-1$ can be interpreted as the hybrid $k$-block path. Table \ref{tableRbCo} presents the summary statistics of the distribution of $k_o$ (over 100 sequences) for each $r_B$. We can see the values of minimum, first quartile, median, third quartile and maximum in each distribution and these indicate how much $k_o$ varies. Observe the difference with the previous example in Subsection \ref{sec:example2} -- the variation of $C_o$ is tremendous and for $r_B=0.4,0.6,0.8$, the maximum value of $k_o$ is much larger than the sequence length 1000 (the number of observation sequences out of 100 for which $k_o$ is larger than 1000 is 3, 3 and 2, respectively). In particular, $k_o$ might be even more than 9000. This contradicts the naive intuition that when $C>n$, then every hybrid path should be the Viterbi one, {because we have reached the maximum block length $n$}.\\\\
\begin{table}[htbp!]
\begin{center}
{
\begin{tabular}{|c||c|c|c|c|c|}
  \hline
$r_B$ & Min & $Q_1$ &  Median  & $Q_3$  &   Max  \\
\hline
0.2 &   2  &    7.0  &   14.0   &   39.75 & 353  \\
0.4 &    3  &  10.0  &  23.0  &   48.5 &  2303   \\
0.6 &    5  &    16.0  &  28.0  &   91.5 & 1782  \\
0.8 &   6 &    28.0 &   52.5 &    99.5 & 9193   \\
\hline
\end{tabular}}
\end{center}
\caption{\label{tableRbCo} The distribution of the smallest integer $k_o$ in Example \ref{sec:example3}, such that for $C\geq k_o$ the hybrid path equals the Viterbi path.}
\end{table}
\section*{Appendix: Alternative parametrization of model \eqref{milano}.}
Consider the related Markov chain model in Subsection \ref{sec:milano}  $P_X=P_Y$. Reparametrize the transition matrix
\[
\mathbb{P} =
\bordermatrix{ ~ & (1,a) & (1,b) & (2,a) & (2,b)\cr
     (1,a) & p\lambda_1 & p(1-\lambda_1) & p(1-\lambda_1) & 1+p\lambda_1 -2p\cr
     (1,b) & p\lambda_2 & p(1-\lambda_2) & q-p\lambda_2 & 1+p \lambda_2 -q-p\cr
     (2,a) & q\mu_1 & q(1-\mu_1) & p-q\mu_1 & 1+q\mu_1-p-q \cr
     (2,b) & q\mu_2 & q(1-\mu_2) & q(1-\mu_2) & 1+q\mu_2-2q\cr
    }
\]
as follows:
\[
\mathbb{P} =
\bordermatrix{ ~ & (1,a) & (1,b) & (2,a) & (2,b)\cr
     (1,a) & p\lambda_1 & p(1-\lambda_1) & (1-p) \theta_1 & (1-p)(1-\theta_1) \cr
     (1,b) & p\lambda_2 & p(1-\lambda_2) & (1-p) \theta_2 & (1-p)(1-\theta_2) \cr
     (2,a) & q\mu_1 & q(1-\mu_1) & (1-q)\rho_1 & (1-q)(1-\rho_1) \cr
     (2,b) & q\mu_2 & q(1-\mu_2) & (1-q)\rho_2 & (1-q)(1-\rho_2) \cr
    },
\]
where
\[ \theta_1= {p(1-\lambda_1)\over 1-p},\quad \theta_2={q-p\lambda_2 \over 1-p}, \quad \rho_1={p-q\mu_1 \over 1-q},\quad \rho_2={q(1-\mu_2)\over 1-q}. \]
Thus, for given $p$ and $q$, the new parameters $\theta_1$, $\theta_2$, $\rho_1$ and $\rho_2$ are functions of $\lambda_1$, $\lambda_2$, $\mu_1$ and  $\mu_2$, respectively. The parameters represent the following probabilities:
\[ P(Y_2=a|Y_1=a,X_1=1,X_2=1)=\lambda_1,\quad P(Y_2=b|Y_1=a,X_1=1,X_2=1)=1-\lambda_1,  \]
\[ P(Y_2=a|Y_1=b,X_1=1,X_2=1)=\lambda_2,\quad P(Y_2=b|Y_1=b,X_1=1,X_2=1)=1-\lambda_2,  \]
\[ P(Y_2=a|Y_1=a,X_1=1,X_2=2)=\theta_1,\quad P(Y_2=b|Y_1=a,X_1=1,X_2=2)=1-\theta_1,  \]
\[ P(Y_2=a|Y_1=b,X_1=1,X_2=2)=\theta_2,\quad P(Y_2=b|Y_1=b,X_1=1,X_2=2)=1-\theta_2,  \]
\[ P(Y_2=a|Y_1=a,X_1=2,X_2=1)=\mu_1,\quad P(Y_2=b|Y_1=a,X_1=2,X_2=1)=1-\mu_1,  \]
\[ P(Y_2=a|Y_1=b,X_1=2,X_2=1)=\mu_2,\quad P(Y_2=b|Y_1=b,X_1=2,X_2=1)=1-\mu_2,  \]
\[ P(Y_2=a|Y_1=a,X_1=2,X_2=2)=\rho_1,\quad P(Y_2=b|Y_1=a,X_1=2,X_2=2)=1-\rho_1,  \]
\[ P(Y_2=a|Y_1=b,X_1=2,X_2=2)=\rho_2,\quad P(Y_2=b|Y_1=b,X_1=2,X_2=2)=1-\rho_2.  \]
Thus, given an observation sequence $x^n$, the probability of any state sequence $y^n$ is determined by the initial distribution and the transition probabilities above (representing four transition matrices).
\subsection*{Acknowledgments} This work is supported by the Estonian Research Council grant PRG865.
\subsection*{Conflict of interest} The authors have no conflicts of interest to declare that are relevant to the content of this article.
\bibliographystyle{apa}

\end{document}